\documentclass[11pt]{article}
\usepackage[margin=1in]{geometry} 
\usepackage{graphicx} % Required for inserting images
\usepackage[dvipsnames]{xcolor}
 \usepackage[section]{placeins}
 \usepackage{hyperref}
\hypersetup{
  colorlinks=true,
  linkcolor=black,
  filecolor=black,
  urlcolor=black,
  citecolor=black,
  pdftitle={Your Document Title},
  pdfpagemode=FullScreen,
}
\usepackage{enumitem}
\usepackage{multirow}
\usepackage{tabularray}
\usepackage{amsmath}
\usepackage{amsfonts}
\usepackage{amssymb}
\usepackage{amsthm}
\usepackage{graphicx}
\usepackage{xcolor}
\usepackage[table]{xcolor}
\usepackage{colortbl} 
\usepackage{setspace}
\usepackage{physics}
\usepackage{authblk}
\usepackage{tikz}
\usepackage{caption}
\usepackage{mathdots}
\usepackage{yhmath}
\usepackage{mathtools}
\usepackage{cancel}
\usepackage{color}
\usepackage{siunitx}
\usepackage{array}
\usepackage{multirow}
\usepackage{longtable}
\usepackage{tabularx}
\usepackage{extarrows}
\usepackage{booktabs}
\usetikzlibrary{fadings}
\usetikzlibrary{patterns}
\usetikzlibrary{shadows.blur}
\usetikzlibrary{shapes}
\usepackage[normalem]{ulem}
\usepackage{framed}
\usepackage{lipsum}
\usepackage{algorithm}
\usepackage{algpseudocode}
\usepackage{cleveref}

\newcommand{\independent}{\perp \!\!\! \perp}

\DeclareMathOperator*{\argmin}{argmin}
\DeclareMathOperator*{\argmax}{argmax}

\newtheorem{assumption}{Assumption}
\newtheorem{condition}{Condition}
\newtheorem{theorem}{Theorem}

\setlist[itemize]{itemsep=1mm,parsep=1mm,topsep=0mm}
\setlist[enumerate]{itemsep=1mm,parsep=1mm,topsep=0mm}
\usepackage[numbers,sort&compress]{natbib}

\title{Targeted maximum likelihood estimation of vaccine effectiveness and immune correlates in test-negative design studies with missing data}

\author[1]{Leah I. B. Andrews*}
\author[2]{Lars van der Laan}
\author[1,3,4]{Peter B. Gilbert}
\affil[1]{Department of Biostatistics, University of Washington}
\affil[2]{Department of Statistics, University of Washington}
\affil[3]{Vaccine and Infectious Disease Division, Fred Hutchinson Cancer Center}
\affil[4]{Public Health Sciences Division, Fred Hutchinson Cancer Center}

\date{}

\begin{document}

\maketitle
% Abstract
\begin{abstract}
The test-negative design (TND) is a resource-efficient observational study design that can assess vaccine effectiveness and exposure-proximal immune correlates of disease. The TND enrolls symptomatic individuals seeking diagnostic testing and compares case status by an exposure variable, such as vaccination status or immune marker level, that is measured at testing. While the TND reduces confounding by healthcare-seeking behavior, other sources of confounding may remain. TND studies may also have missing data in the exposure variable due to incomplete records or two-phase sampling designs. We present a targeted maximum likelihood estimation approach involving a semiparametric logistic regression model that targets a causal conditional risk ratio of symptomatic disease in the healthcare-seeking population. Under causal and missing at random assumptions, our method produces an efficient, asymptotically linear estimator that provides flexible, data-driven confounding control and valid causal inference when analyzing TND studies with missing exposure variable data. We evaluate our method's finite sample properties using plasmode simulations of a two-phase TND immune correlates study. We also apply our method to assess COVID-19 vaccine effectiveness and antibody marker correlates of COVID-19 from TND study cohorts derived from the Moderna Coronavirus Efficacy phase 3 trial.
\end{abstract}

% Keywords
\noindent \textbf{Keywords:} test-negative design, targeted maximum likelihood estimation, vaccine effectiveness, immune correlates, missing data, two-phase sampling
\\

\noindent
*\textbf{Corresponding author}: Leah I. B. Andrews (email address: landrew2@uw.edu, ORCID iD: 0000-0002-3418-9468).

\section{Introduction}
\label{sec:intro}

The test-negative design (TND) is a practical observational study design that has been used to monitor post-marketing vaccine effectiveness against pneumococcal disease, influenza, rotavirus, cholera, pertussis, and most recently, COVID-19 \cite{broome_pneumococcal_1980, chua_use_2020,patel_evaluation_2021, chang_effectiveness_2022, tsang_prior_2024}. The TND enrolls individuals who obtain diagnostic testing for a target pathogen and meet a symptom definition. Vaccine effectiveness is assessed by comparing vaccination status between individuals who test positive (cases) and individuals who test negative (noncases), after adjusting for potential confounders \cite{jackson_test-negative_2013, sullivan_theoretical_2016}. TND studies are resource-efficient because they recruit cases and noncases identically, collect participants' information in one visit, and obtain a high proportion of cases. The TND also reduces bias from healthcare-seeking behavior (i.e., seeking medical attention when ill). Individuals with healthcare-seeking behavior are more likely to obtain frequent or intensive healthcare services that increase their probability of being diagnosed, engaging in risk-reducing behaviors (e.g., hand-washing, vaccination), and participating in studies, which induces bias in observational studies since healthcare-seeking behavior is difficult to measure \cite{ sullivan_theoretical_2016,clemens_resolving_1984}. The TND reduces this bias by limiting enrollment to symptomatic individuals who obtain testing and assuming they have similar healthcare-seeking behavior \cite{jackson_test-negative_2013,sullivan_theoretical_2016}.  

A novel application of the TND is to study exposure-proximal immune correlates of disease (i.e., studying the association of immune markers present at the time of potential pathogen exposure with a clinical disease endpoint) \cite{follmann_immune_2022,sumner_antisars-cov-2_2024,zhang_use_2024, middleton_statistical_2025,follmann_test-negative_2025}. Exposure-proximal immune correlates can provide greater insights about immunological mechanisms of protection than immune markers measured at fixed time points and can help validate surrogate endpoints for inferring vaccine efficacy. In a TND immune correlates study, a participant contributes their specimen sample (which can be assayed later for immune markers), diagnostic test, and additional covariates in a single visit. To interpret immune markers as exposure-proximal immune correlates, TND studies must either assume that participants' immune marker levels are identical at testing and potential pathogen exposure or account for rapid memory immune responses following natural infection \cite{follmann_immune_2022, follmann_test-negative_2025}. The TND is especially advantageous for studying rare severe outcomes and cell-mediated immune markers that are rarely studied in randomized, controlled clinical trials (RCTs) and prospective cohort studies because their measurement requires resource-intensive storage of peripheral blood mononuclear cells \cite{follmann_test-negative_2025,gilbert_immune_2022,benkeser_immune_2023,fong_immune_2022,fong_immune_2023,zhang_humoral_2022,hertoghs_vaccine-induced_2025, krause_making_2022}. 

Despite its advantages, the TND is an observational study design that is subject to confounding from characteristics like age, comorbidities, SARS-CoV-2 infection history, vaccination history, region, and calendar date \cite{ sullivan_theoretical_2016, chang_effectiveness_2022, chua_use_2020,dean_temporal_2020}. To account for confounding, most TND studies fit an ordinary logistic regression adjusted for confounders \cite{bond_regression_2016, thompson_effectiveness_2021, lopez_bernal_effectiveness_2021,chung_effectiveness_2021}, which assumes a parametric model that may induce bias if misspecified. Alternatively, some TND studies match by covariates and fit a conditional logistic regression \cite{bond_regression_2016,ranzani_vaccine_2022,pilishvili_effectiveness_2021}, which may not adequately control for confounding, may be less efficient than covariate adjustment through regression, and is more resource-intensive when studying multiple symptom definitions \cite{rothman_modern_1998,rose_why_2009}. Only recently have TND methods incorporated machine-learning methods to provide flexible confounding control \cite{jiang_double_2025}. 

TND studies may also have missing data in the exposure variable, which must be considered to ensure generalizable results. Missing vaccination status may arise in a TND vaccine effectiveness study from lack of vaccine records or unverified self-report \cite{tenforde_effectiveness_2021,olson_effectiveness_2021, patel_evaluation_2021}. Most TND studies require known vaccination status in the study eligibility criteria and conduct complete-case analyses \cite{tenforde_effectiveness_2021, ranzani_effectiveness_2021}, which may induce bias if individuals missing vaccination status have different characteristics than those with vaccination status observed \cite{tsiatis_models_2006}. Since measuring immune markers is expensive and labor-intensive, TND immune correlates studies may apply two-phase sampling, a common cost-effective and efficient design in immune correlate studies \cite{white_two_1982, rothman_modern_1998,breslow_using_2009,gilbert_controlled_2022, kenny_immune_2024}. 
Under this design, demographic and clinical information, diagnostic tests, and specimen samples are collected from all enrolled individuals at testing (i.e., phase one). Once test results are available, immune markers are measured from a subset of phase one individuals sampled into phase two with sampling probabilities dependent on case status and covariates. Although the missingness framework is known, analyses must account for the missingness patterns to obtain an unbiased estimate. 

With increasing reliance on the TND for post-marketing vaccine evaluation, it is essential to develop more robust statistical methods for TND analyses than ordinary logistic regression. Targeted maximum likelihood estimation approaches use semiparametric theory and machine-learning algorithms to produce less biased estimates under fewer assumptions than parametric models \cite{van_der_laan_targeted_2006,van_der_laan_targeted_2011}. These approaches also respect model constraints and may have better finite sample properties than one-step estimation and estimating equation approaches \cite{kang_demystifying_2007, porter_relative_2011}. van der Laan and Gilbert \cite{van_der_laan_semiparametric_2025} developed a targeted maximum likelihood estimator (TMLE) to study strain-specific relative COVID-19 vaccine efficacy in case-only designs.  This method assumes a semiparametric logistic regression model that allows for flexible, data-driven confounding adjustment and missing cause of failure information (e.g., viral genotype) in test-positive cases. 

In this work, we adapt the TMLE developed by van der Laan and Gilbert \cite{van_der_laan_semiparametric_2025} to a TND setting to study vaccine effectiveness and exposure-proximal immune correlates. We extend the identifying assumptions to target a causal conditional risk ratio of symptomatic disease in the healthcare-seeking population and account for missing data in the exposure variable (vaccination status or exposure-proximal immune marker level). Our method flexibly adjusts for confounding and produces an efficient, asymptotically linear substitution estimator that provides valid inference in TND settings. We compare the performance of our TMLE approach to pseudo-likelihood logistic regression and ordinary logistic regression in two-phase TND immune correlates study simulations. Using data from the Moderna Coronavirus Efficacy (COVE) trial, a phase 3 RCT of the mRNA-1273 vaccine in the United States (US) \cite{el_sahly_efficacy_2021}, we evaluate COVID-19 vaccine effectiveness and antibody marker correlates of COVID-19 using the TMLE and comparison methods.

\section{Estimands}
\label{sec:estimands}

In this section, we define our target estimand, the full data causal conditional risk ratio, and the observed conditional odds ratio from the TND data. We first provide assumptions to identify the full data conditional risk ratio. We then provide additional assumptions to identify the full data causal conditional risk ratio. 

\subsection{Causal estimand and ideal full data structure}
\label{subsec:causalq}
We are interested in assessing the effect of a binary exposure variable (vaccination status or exposure-proximal immune marker level) on symptomatic disease in a healthcare-seeking population. Relevant vaccine regimens are completion of primary series vaccination within a time period vs. no vaccination, different vaccine types, or number of vaccine doses. Relevant binary immune marker levels for comparison are high vs. low immune markers according to some threshold, like the limit of detection. Given our COVE motivating example, we discuss both exposures, vaccination status and immune marker level, against symptomatic COVID-19. 

Our population of interest $\mathcal{P}$ is the healthcare-seeking population ($U=1$), which consists of independent and identically distributed individuals who would seek care if ill and have unobserved full data structure $O_F = (A, Y,W,D,X)\sim P_F \in \mathcal{P}$ (Figure \ref{fig:dag1}). We define $A$ as a binary exposure variable, either vaccination status at time of potential SARS-CoV-2 exposure or exposure-proximal immune marker level. $Y$, $W$, and $D$ are indicators of SARS-CoV-2 infection, a cause other than SARS-CoV-2 that could induce symptoms, and meeting the symptom definition, respectively. $X\in\mathbb{R}^r$ represents covariates, like demographic and clinical characteristics, necessary to satisfy the assumptions listed in Sections \ref{subsec:statparam} and \ref{subsec:causalparam}.

An individual has causal data structure $O_{F,ca}=(O_F, Y(0),Y(1), D(0),D(1))$ $\sim P_{F,ca}$, where $Y(a)$ represents the potential outcome of SARS-CoV-2 infection and $D(a)$ represents the potential outcome of meeting the symptom definition under exposure status $a$, for $a=0,1$. Individuals have all potential outcomes but only one set of observed values $\{Y, D\}$. We target a full data causal conditional risk ratio,
\begin{equation}
RR (P_{F,ca}(x)):=\frac{P_{F,ca}(Y(1)=1, D(1)=1|X=x)}{P_{F,ca}(Y(0)=1,D(0)=1| X=x)},
\label{eqn:causalrr}
\end{equation}
which represents the reduction in COVID-19 risk comparing individuals with $A=1$ versus $A=0$ in the healthcare-seeking population who share additional characteristics $X=x$. We limit inference to the healthcare-seeking population to avoid inducing selection bias, which is more likely to occur for endpoints measuring a milder form of disease \cite{jackson_test-negative_2013,sullivan_theoretical_2016,shi_comparison_2017,lewnard_theoretical_2021}. 

\subsection{Observed test-negative design data}
\label{subsec:obsdata}
The observed TND data structure is $O=DS(\Delta, \Delta A,Y, X)\sim P$ and a subset of $O_F$ (Figure \ref{fig:dag1}). Individuals are enrolled in a TND if they obtain SARS-CoV-2 testing within some testing window ($S=1$) and meet the symptom definition ($D=1$). Individuals are enrolled before $Y$ is measured with SARS-CoV-2 diagnostic tests. $W$ is a latent variable that affects $D$ in the population but is not typically observed in the TND. $X$ is collected through electronic health records, medical examination, specimen collection, and/or self-report. $\Delta$ is an indicator of observing $A$, which may be missing unintentionally or by two-phase sampling design. Cases are individuals who meet the symptom definition, obtain SARS-CoV-2 testing, and test SARS-CoV-2 positive ($D= S= Y=1$). Noncases meet the symptom definition, obtain SARS-CoV-2 testing, and test SARS-CoV-2 negative ($D= S=1, Y=0$). Since all TND individuals have symptoms and obtain SARS-CoV-2 testing, we assume they have identical healthcare-seeking behavior ($U=1$) and are representative of the healthcare-seeking population. 

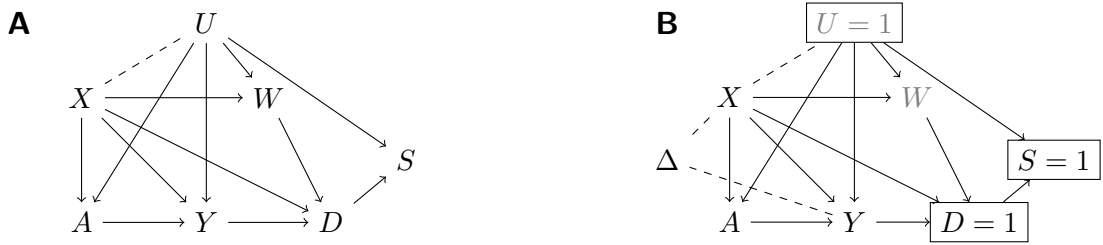
\begin{figure}[ht]
  \begin{minipage}[c]{0.36\textwidth}
    \centering
% This code uses the tikz package
\begin{tikzpicture}[scale = 3.3]
\node (v0) at (-0.75, 1.2) {$U$};
\node (v1) at (0.05,0.65) {$S$};
\node (v2) at (-1.25,0.4) {$A$ };
\node (v3) at (-0.75,0.4) {$Y$};
\node (v4) at (-0.25,0.4) {$D$};
\node (v5) at (-1.25,.9) {$X$};
\node (v7) at (-0.5, .9) {$W$};
\node [font=\sffamily](v8) at (-1.5, 1.2) {\large \textbf{A}};
\draw [->] (v0) edge (v3);
\draw [->] (v0) edge (v2);
\draw [->] (v0) edge (v7);
\draw [->] (v2) edge (v3);
\draw [dashed] (v5) edge (v0);
\draw [->] (v3) edge (v4);
\draw [->] (v4) edge (v1);
\draw [->] (v5) edge (v2);
\draw [->] (v5) edge (v7);
\draw [->] (v5) edge (v3);
\draw [->] (v5) edge (v4);
\draw [->] (v0) edge (v1);
\draw [->] (v7) edge (v4);
\end{tikzpicture}
\end{minipage}
\hfill
  \begin{minipage}[c]{0.58\textwidth}
        \centering
\begin{tikzpicture}[scale = 3.3]
\node [draw, text=gray ](v0) at (-0.75, 1.2) {$U=1$};
\node [draw](v1) at (0.05,0.65) {$S=1$};
\node (v2) at (-1.25,0.4) {$A$ };
\node (v3) at (-0.75,0.4) {$Y$};
\node [draw](v4) at (-0.25,0.4) {$D=1$};
\node (v5) at (-1.25,.9) {$X$};
\node (v6) at (-1.5,0.65) {$\Delta$};
\node [text=gray ](v7) at (-0.5, .9) {$W$};
\node [font=\sffamily](v8) at (-1.5, 1.2) {\large \textbf{B}};
\draw [->] (v0) edge (v3);
\draw [->] (v0) edge (v2);
\draw [->] (v0) edge (v7);
\draw [->] (v2) edge (v3);
\draw [dashed] (v5) edge (v0);
\draw [->] (v3) edge (v4);
\draw [->] (v4) edge (v1);
\draw [->] (v5) edge (v2);
\draw [->] (v5) edge (v7);
\draw [->] (v5) edge (v3);
\draw [->] (v5) edge (v4);
\draw [dashed] (v5) edge (v6);
\draw [dashed] (v3) edge (v6);
\draw [->] (v0) edge (v1);
\draw [->] (v7) edge (v4);
\end{tikzpicture}
  \end{minipage}
\caption[Directed acyclic graph (DAG) of causal relationships in the general population (\textbf{A}) and a test-negative design study (\textbf{B}).]{Directed acyclic graph (DAG) of causal relationships in the general population (\textbf{A}) and a test-negative design study (\textbf{B}) between vaccination status or immune marker level $A$, SARS-CoV-2 infection $Y$, a cause other than SARS-CoV-2 that could induce symptoms $W$, symptoms $D$, covariates $X$, healthcare-seeking behavior $U$, obtaining SARS-CoV-2 testing $S$, and observing vaccination status or immune marker level $\Delta$. Solid arrows indicate causal relationships and dashed lines indicate associations. $U$ and $W$ are not typically observed in test-negative design studies and are in gray. The boxes in DAG \textbf{B} reflect how the test-negative design restricts to symptomatic individuals who obtain SARS-CoV-2 testing, thereby assuming they share identical healthcare-seeking behavior.} 
  \label{fig:dag1}
\end{figure}

\subsection{Identification of the full data conditional risk ratio}
\label{subsec:statparam}

Since the TND uses outcome-dependent sampling, we can obtain
the probability of exposure status given case status and other covariates,
\begin{equation}
\mu_{P}(y,x) := P(A=1|D=1,S=1,\Delta = 1, Y=y, X=x),
\label{eqn:mu}
\end{equation}
from the observed TND data but not the probability of case status given exposure status and other covariates \cite{westreich_berksons_2012}. Thus, the observed conditional odds ratio,
\begin{equation}
OR(P)(x):=\frac{\mu_P(1,x)/(1-\mu_P(1,x))}{\mu_P(0,x)/(1-\mu_P(0,x))},
\label{eqn:or}
\end{equation}
is identified by the observed TND data. However, we introduce the following nonparametric identifying assumptions to identify the full data conditional risk ratio $RR(P_F)(x)$ from 
$OR(P)(x)$ \cite{van_der_laan_semiparametric_2025,jackson_test-negative_2013,broome_pneumococcal_1980,schnitzer_estimands_2022, jiang_double_2025}:

\begin{assumption}
\textit{Weak overlap for symptoms from SARS-CoV-2 and other causes}\\
$P\left( P(Y= 1,D=1| S=1, X) > 0\right)=1.$

\noindent $P\left(P(Y= 0,D=1| S=1, X) > 0\right)=1.$
\label{ass:posy}
\end{assumption}

\begin{assumption}\textit{Weak overlap for exposure status missingness} \\
$P\left(P(\Delta = 1| D=1,S=1,Y,X) > 0\right)=1.$ 
\label{ass:posdelta}
\end{assumption}

\begin{assumption} \textit{Weak overlap for SARS-CoV-2 testing missingness} \\
$P_F\left(P_F(S = 1| D=1,Y, X) >0\right)=1.$
\label{ass:poss}
\end{assumption}

\begin{assumption} \textit{Exposure status is missing at random in the symptomatic population} \\
$S \independent A |D=1, Y, X. $
\label{ass:condsa}  
\end{assumption}

\begin{assumption}\textit{Exposure status is missing at random among TND participants} \\
$\Delta \independent A | D=1, S=1,Y, X.$ 
\label{ass:conddeltaa}  
\end{assumption}

\begin{assumption} \textit{Noncase Exchangeability (core TND assumption)}\\
$ P_F (Y=0,D=1|  A= 1,X)=P_F (Y=0,D=1|A=0,X).$
\label{ass:coretnd}
\end{assumption}

\begin{assumption}
    \textit{Exposure status at SARS-CoV-2 testing accurately reflects exposure status $A$ at time of potential SARS-CoV-2 exposure in TND participants}
    \label{ass:amisclass}
\end{assumption}

\begin{assumption}
    \textit{SARS-CoV-2 test result accurately measures SARS-CoV-2 infection status $Y$ in TND participants}
    \label{ass:ymisclass}
\end{assumption}

\begin{theorem} Identification of full data conditional risk ratio

Given the identifying Assumptions \ref{ass:posy}-\ref{ass:ymisclass}, the observed conditional odds ratio $OR(P)(x)$ equals the full data conditional risk ratio,
  \begin{align}
RR(P_F)(x) &:=\frac{P_F(Y=1, D=1|A=1,X=x)}{P_F(Y=1,D=1|A=0, X=x)}.
  \end{align}
 \label{thm:fulldatarr}
\end{theorem}
Assumption \ref{ass:posy} and \ref{ass:poss} require that TND studies are conducted during seasons and in regions where SARS-CoV-2 and other causes of COVID-19-like symptoms (e.g., influenza, rhinovirus, seasonal allergies) are present and SARS-CoV-2 testing is available. Assumption \ref{ass:posdelta} posits that every $X=x$ subgroup that meets the symptom definition and obtains SARS-CoV-2 testing has some probability of having their exposure status observed.

Assumption \ref{ass:condsa} requires that obtaining a SARS-CoV-2 test at a TND study site is conditionally independent of exposure status, given meeting the symptom definition, SARS-CoV-2 infection status, and covariates. Limiting inference to a healthcare-seeking population, adjusting for covariates associated with testing behavior, and studying a severe COVID-19 outcome could help satisfy this assumption \cite{jackson_test-negative_2013,sullivan_theoretical_2016,lewnard_theoretical_2021,shi_comparison_2017}. To satisfy Assumption \ref{ass:conddeltaa}, TND studies should adjust for covariates related to vaccination status missingness or the two-phase sampling design.

Assumption \ref{ass:coretnd} requires that the incidence of COVID-19-like symptoms not caused by SARS-CoV-2 is the same in both exposure status groups in the healthcare-seeking population with the same covariates. This is satisfied if the exposure status does not affect other causes of COVID-19-like symptoms and confounders of exposure status and other causes of COVID-19-like symptoms (e.g., comorbidities, age, vaccination against other respiratory pathogens) are accounted for \cite{doll_effects_2022, jiang_double_2025, payne_impact_2023}. 

Assumptions \ref{ass:amisclass} and \ref{ass:ymisclass} require that the exposure status and SARS-CoV-2 test result collected at testing accurately measure exposure status at the time of potential SARS-CoV-2 exposure and SARS-CoV-2 infection, respectively.

\subsection{Identification of the full data causal conditional risk ratio}
%Causal Parameter
\label{subsec:causalparam}
To identify the full data causal conditional risk ratio, we require standard causal assumptions \cite{robins_causal_2024,van_der_laan_semiparametric_2025,schnitzer_estimands_2022}: 

\begin{assumption} \textit{Consistency}\\
$\{Y(a), D(a)|A=a \} \overset{d}{=} \{Y, D | A=a\}$ for $ a=\{0,1\}.$
\label{ass:cons}
\end{assumption}
\begin{assumption}
 \textit{No interference}\\
$Y_i(a),D_i(a)\independent A_k$ for $a=\{0,1\}$ and all pairs of individuals $i\neq k$.
\label{ass:inter}
\end{assumption}
\begin{assumption}
\textit{No unmeasured confounding}\\
 $ Y(a), D(a) \independent A|X$ for $a=\{0,1\}.$
 \label{ass:confound}
\end{assumption}

\begin{theorem} Identification of full data causal conditional risk ratio

Given the identifying Assumptions \ref{ass:posy}-\ref{ass:confound}, the observed conditional odds ratio $OR(P)(x)$ equals the full data causal conditional risk ratio $RR(P_{F,ca})(x)$ .
    \label{thm:causalrr}
\end{theorem}

Consistency requires no variability within exposure statuses that may affect the potential outcomes and that all individuals can achieve both exposure statuses \cite{rubin_causal_2005, hudgens_toward_2008}. TND studies should carefully define the population of inference and exposure statuses and account for additional variables, like vaccination history and prior infection, to satisfy this assumption. To minimize interference  \cite{cox_planning_1958,rubin_randomization_1980}, TND studies could enroll individuals at many sites, limit one individual per household, or incorporate methods that account for interference \cite{hudgens_toward_2008, morgan_social_2013,schnitzer_estimands_2022}. To satisfy no unmeasured confounding, TND studies should measure and adjust for characteristics that may affect individuals’ exposure status, SARS-CoV-2 infection, and/or symptom development.

\section{Semiparametric estimation and inference}
\label{sec:semi}
In this section, we describe a semiparametric model to construct an efficient but asymptotically biased substitution estimator of $OR(P)(x)$. We then update this estimator using targeted maximum likelihood estimation to obtain a consistent, asymptotically linear, debiased TMLE.

\subsection{Partially linear logistic regression model}
\label{subsec:semiass}
To obtain an expression for $OR(P)(x)$, we impose a partially linear logistic regression model for $P_F(A=1|D=1,Y=y,  X=x)$ \cite{tchetgen_tchetgen_doubly_2010}:
\begin{equation}\text{logit} \{P_F(A=1|D=1,Y=y,  X=x)\} = y\beta(P_F)^T \underline{f}(x) + h_{P_F}(x) 
\label{eqn:semilogit}
\end{equation} for an unknown vector of coefficients $\beta(P_F)\in\mathbb{R}^b$, a known vector-valued function $\underline{f}:\mathbb{R}^r\to\mathbb{R}^b$, and an unspecified function $h_{P_F}(x):=\text{logit} \{P_F(A=1|D=1,Y=0,  X=x)\}$. This semiparametric model for $P_F$ provides an interpretable relationship between exposure status and case status that can incorporate effect modification, $\exp\left(\beta(P_F)^T \underline{f}(x)\right)=OR(P_F)(x)$, and nonparametrically adjusts for covariates to help satisfy the assumptions in Sections \ref{subsec:statparam} and \ref{subsec:causalparam}. The partially linear logistic regression model is advantageous because it provides the same useful interpretation as the ordinary logistic regression model without having to specify the relationships between covariates $X$ and exposure $A$.

\subsection{Plug-in estimator}
\label{subsec:plugin}

Let $\mathcal{M}$ be a collection of distributions that satisfy $\mu_{P'}(y,x)=\text{expit}\{y\beta(P')^{T} \underline{f}(x) + h_{P'}(x)\},$ where $\beta(P')\in\mathbb{R}^b$ and $h_{P'}(w,t):\mathbb{R}^r\to\mathbb{R}^b$ for any $P'\in\mathcal{M}$. Our observed TND data has sample size $n$ and is generated by $P\in\mathcal{M}$. Thus, $P$ must satisfy $\log OR(P)(x)=\beta(P)^T\underline{f}(x)$ for $\beta(P)\in\mathbb{R}^b$.

To construct the efficient influence function of $\beta(P')$ for $P'\in\mathcal{P}$, we define the conditional probability of COVID-19 as $\widetilde{\pi}_{P'}(x) := P'(Y=1,D=1\mid S=1, \Delta=1, X=x)$, the conditional variance as $\sigma^2_{P'}(y,x):=\mu_{P'}(y,x)(1-\mu_{P'}(y,x))$, the invertible scaling matrix as
\begin{equation*}
    {\Lambda}_{P'}^{-1} := E_{P'}\left[\underline{f}(X)\underline{f}(X)^T\frac{\Delta(1-\widetilde{\pi}_{P'}(X))\widetilde{\pi}_{P'}(X)\sigma^2_{P'}(1,X)\sigma^2_{P'}(0,X)}{ (1-\widetilde{\pi}_{P'}(X))\sigma^2_{P'}(0,X) +\widetilde{\pi}_{P'}(X)\sigma^2_{P'}(1,X) }\right],
\end{equation*} and
\begin{equation*}
H_{P'}(y,x):=\left(y-\frac{E_{P'}\left[Y\sigma^2_{P'}(1,x)|\Delta=1,X=x\right]}{E_{P'}\left[\sigma_{P'}^2(Y,x)|\Delta=1,X=x\right]}\right).
\end{equation*} 

\begin{theorem} Efficient influence function of $\beta(P)$

The vector-valued efficient influence function of $\beta(P)$ for $P\in\mathcal{M}$ is 
\begin{equation} D_{P}(o) :=   (\Lambda_{P}\underline{f}(x))\circ \delta H_{P}(y,x)  \left[ \delta a  - \mu_{P}(y,x)\right], \label{eqn:eif} \end{equation}
where $\circ$ represents the coordinate-wise product. \label{thm:eif}
\end{theorem}
We obtain this efficient influence function by adapting the efficient influence function from the partially linear logistic regression model coefficient \cite{van_der_laan_readings_2009,tchetgen_tchetgen_doubly_2010, van_der_laan_semiparametric_2025} and conditioning on $D=S=\Delta=1$.

We estimate $\beta(P)$ and $h_P$ in the observed TND data, such that
\begin{align} 
\{ \beta(P_n),  h_{P_n}\}:= \argmin_{\beta(P) \in \mathbb{R}^b, h_P \in \mathcal{H}} \frac{1}{n}\sum_{i=1}^n& S_i \Delta_i \left\{ \log\left\{1 + \exp\left(Y_i\beta(P)^T \underline{f}(X_i) + h_P(X_i) \right) \right\}\right.  \nonumber \\
&\left. - A_i \left(Y_i\beta(P)^T \underline{f}(X_i) + h_P(X_i) \right) \right\},
\end{align}
where $\mathcal{H}$ is a class of nuisance functions, such as those with bounded variation. Minimizing this risk is equivalent to fitting a partially linear logistic regression of $A$ on $Y$ and $X$ using individuals in the TND data with complete vaccination or immune marker information. 

Machine-learning algorithms that meet the partially linear logistic model constraint, such as generalized additive models \cite{hastie_generalized_1987}, smoothing splines \cite{friedman_multivariate_1991}, highly adaptive lasso \cite{benkeser_highly_2016}, or ensemble methods, like super learner \cite{van_der_laan_super_2007}, can be used for estimation to reduce bias from model misspecification. Since model performance varies by dataset, ensemble methods are advantageous because they can consider many machine-learning algorithms for estimation and can achieve the same performance as the optimal estimator for a dataset when a rich and diverse library of parametric and nonparametric machine-learning algorithms is considered \cite{van_der_laan_unified_2003-1,dudoit_asymptotics_2005,van_der_vaart_oracle_2006, van_der_laan_super_2007}.  

The resulting plug-in estimator, $ \mu_{P_n}(y,x) := \text{expit} \left\{  y\beta(P_{n})^T \underline{f}(x) +  h_{P_n}(x)\right\}$, respects the model constraints of $\mathcal{M}$ and is a plug-in estimator that can be used to estimate $OR(P)(x)$,
\begin{equation}
OR(P_{n})(x) := \frac{\mu_{P_n}(1,x)(1-\mu_{P_n}(0,x))}{\mu_{P_n}(0,x)(1-\mu_{P_n}(1,x))} = \exp \left\{ \beta(P_{n})^T \underline{f}(x) \right\}.
\label{eqn:estOR}
\end{equation}
However, flexibly estimating the nuisance functions causes bias in $\beta(P_n)$ \cite{bickel_efficient_1993, van_der_laan_semiparametric_2025} and must be addressed. 

\subsection{Targeted maximum likelihood estimation}
\label{subsec:tmle}
We apply targeted maximum likelihood estimation to update and debias $\beta(P_n)$. The TMLE $\beta(P_n^*)$ is constructed by updating $P_{n}$ to $P_n^*$ using an iterative algorithm involving maximum likelihood estimation to solve the efficient score equation $\frac{1}{n}\sum_{i=1}^nD_{P_n^*}(O_i)= 0$.  The algorithm fluctuates $P_n\in\mathcal{M}$ along the least favorable model, as determined by the efficient influence function \cite{van_der_laan_targeted_2011}. Our estimator fluctuates along the logistic fluctuation submodel $(P'(\varepsilon):\varepsilon\in\mathbb{R}^b)\ \subset\mathcal{M}$ and varies $\mu_{P'(\varepsilon)}(y,x)$ through $\varepsilon$ while holding the other components of $P_n$ constant. The submodel satisfies
$$\varepsilon  \mapsto \text{logit}\{\mu_{P'(\varepsilon)}(y,x)\}=\text{logit}\{\mu_{P'}(y,x)\}+\varepsilon^T\underline{f}(x)H_{P'}(y,x): \varepsilon\in\mathbb{R}^b,$$
where $\mu_{P'(\varepsilon)}(y,x)=\text{expit}\{y\beta(P'(\varepsilon))^T\underline{f}(x)+h_{P'(\varepsilon)}(x)\}$, $\beta(P'(\varepsilon))=\beta(P')+\varepsilon$, and $h_{P'(\varepsilon)}(x)=h_{P'}(x)+\varepsilon^T\underline{f}(x)H_{P'}(y=0,x)$, which respect the constraints of $\mathcal{M}$. At each step, maximum likelihood estimation is performed according to the following working log-likelihood function,
$$L_n(P'):=\frac{1}{n}\sum_{i=1}^n\Delta_i\left\{A_i\log\mu_{P'}(Y_i,X_i)+(1-A_i)\log(1-\mu_{P'}(Y_i,X_i))\right\},$$
with score vector,

$$\frac{d}{d\varepsilon}L_n(P'(\varepsilon)):=\frac{1}{n}\sum_{i=1}^n\Delta_iH_{P'}(Y_i,X_i)[A_i-\mu_{P'(\varepsilon)}(Y_i,X_i)]. $$

Formally, at step $k=0$, the initial estimate of $P$ is $P_{n}^{(k=0)} =P_{n}\in\mathcal{M}$, the plug-in estimator in Section \ref{subsec:plugin}. Then, for steps $k\in\mathbb{N}$ , we compute $P_{n}^{k}=P_{n}^{(k-1)}(\varepsilon_n^{(k)})$ and consequently,  $\beta(P_{n}^{k})=\beta(P_{n}^{(k-1)})+\varepsilon_n^{(k)}$, where $\hat{\varepsilon}_n^{(k)}=\argmax_{\varepsilon\in\mathbb{R}^b} L_n(P_{n}^{(k-1)}(\varepsilon))$. The process ends at step $k=K$ when $P_{n}^{K}=P_n^*$ results in $\frac{1}{n}\sum_{i=1}^nD_{P_n^*}(O_i)\approx o_P\left(\frac{1}{\sqrt{n}}\right)$ and $\varepsilon_n^{(K)}$ is nearly the zero vector. The resulting TMLEs are $\beta(P_n^*)$ for $\beta(P)$ and $OR(P_n^*)(x)=\exp \left\{  \beta(P_{n}^{*})^T \underline{f}(x) \right\}$ for $OR(P)(x)$, where $P_n^*\in\mathcal{M}$. 

\subsection{Asymptotic properties}
\label{subsec:asymp}

The TMLEs, $\beta(P_n^*)$ and $OR(P_n^*)(x)$, are substitution estimators and one-step estimators that are efficient and asymptotically linear under the following regularity conditions \cite{bickel_efficient_1993, van_der_laan_semiparametric_2025}:

\begin{condition}
$\underline{f}(X)$ is $P_F$-uniformly bounded and
$E_P \left[ \underline{f} (X) \underline{f}(X)^T \right]$ is invertible. 
\label{cond:inv}
\end{condition}
\begin{condition} There exists some $\epsilon>0$ such that \newline $P(1-\epsilon >\mu_{P_{n}^*}(Y,X) > \epsilon) = P(1-\epsilon>\mu_P(Y,X)>\epsilon)=1$.
\label{cond:bound}
\end{condition}
\begin{condition} $\widetilde{\pi}_{P_n}(X)$ and $\mu_{P_{n}^*}(Y,X)$  are in a uniformly bounded $P$-Donsker function class with probability 1. 
\label{cond:donsker}
\end{condition}
\begin{condition} $\|\widetilde{\pi}_{P_n}(X) - \widetilde{\pi}_P(X)\| = o_P(n^{-1/4})$ and $\max_{y\in\{0,1\}}\|\mu_{P_{n}^*}(Y,X)-\mu_{P}(Y,X) \|= o_P(n^{-1/4})$, where $\|\cdot\|$ is the $L_2(P)$ norm.
\label{cond:nuisrate}
\end{condition}

\begin{theorem} Asymptotic distribution of $\beta(P_{n}^*)$ 

Given $P\in\mathcal{M}$, Assumptions \ref{ass:posy}-\ref{ass:conddeltaa}, \ref{ass:amisclass}, and \ref{ass:ymisclass} and Conditions \ref{cond:inv}-\ref{cond:nuisrate}, the TMLE $\beta(P_{n}^*)$ is efficient and asymptotically linear with influence function $D_P$ and asymptotic distribution
\[\sqrt{n} \left(\beta(P_{n}^*) - \beta(P) \right) \longrightarrow_d N\left(0, \text{cov}(D_P(O)) \right),\]
where $\text{cov}(D_{P}(O)) := E\left[D_{P}(O)D_{P}(O)^T \right]\in \mathbb{R}^b \times \mathbb{R}^b$. The covariance can be consistently estimated using the empirical covariance of the efficient influence function.
\label{thm:asympb}
\end{theorem}

\begin{theorem} Asymptotic distribution of $OR(P_n^*)(x)$

    Under the same assumptions and conditions as Theorem \ref{thm:asympb}, $\log OR(P_n^*)(x)= \beta(P_{n}^{*})^T \underline{f}(x) $ is asymptotically linear, efficient, and asymptotically normal with asymptotic distribution
$$\sqrt{n} \left(\log OR(P_n^*)(x) -\log OR(P)(x) \right) \longrightarrow_d N\left(0, \underline{f}(x)^T\text{cov}(D_P(O)) \underline{f}(x)\right).$$
The asymptotic covariance matrix can be consistently estimated using the empirical covariance of the efficient influence function.
\label{thm:asympor}
\end{theorem}
Condition \ref{cond:inv}  requires the user-defined $\underline{f}(x)$ to be uniformly bounded and have an invertible second moment to uniquely identify $\beta(P)$ from $P$. Condition \ref{cond:bound} bounds the observed TND conditional odds ratio estimand and estimator and is necessary for stability. Conditions \ref{cond:donsker} and \ref{cond:nuisrate} require well-behaved nuisance estimators that converge faster than $n^{-1/4}$. Some machine-learning methods, like random forests and gradient boosting, require cross-fitting \cite{schick_asymptotically_1986, van_der_laan_targeted_2011,chernozhukov_doubledebiased_2018} to satisfy Condition \ref{cond:donsker}. Other machine-learning methods, like smoothing splines \cite{friedman_multivariate_1991,tibshirani_regression_1996}, highly adaptive lasso \cite{benkeser_highly_2016}, and reproducing kernel Hilbert space estimators, satisfy Conditions \ref{cond:donsker} and \ref{cond:nuisrate}. 

Code to implement the TMLE approach in R version 4.4.0 \cite{r_core_team_r_2024} using causalglm \cite{van_der_laan_causalglm_2024}, sl3 \cite{coyle_sl3_2021}, nnls \cite{mullen_nnls_2024}, and hal9001 \cite{hejazi_hal9001_2020,hejazi_hal9001_2020-1} can be found at
\href{https://github.com/leahandrews/semiparametric-tmle-tnd}{https://github.com/leahandrews/semiparametric-tmle-tnd}. 

\section{Simulation study}
\label{sec:sim}
We conducted a plasmode simulation study using COVE participant characteristics \cite{el_sahly_efficacy_2021} to compare empirical performance of our TMLE method, a standard TND method, and a two-phase cohort method when analyzing data from a two-phase TND immune correlates study. We evaluated inferences on  $\beta(P_F)$, or the effect of having a high immune marker level on COVID-19 after adjusting for additional covariates, under three effect sizes, three confounding settings, three two-phase sampling designs, and four phase one sample sizes over 1000 Monte Carlo repetitions.

\subsection{Data-generating mechanism and estimators}
\label{subsec:datagen}

For each simulated dataset, we generated a healthcare-seeking population of $N=50,000$ individuals according to Figure \ref{fig:dag1}. We sampled per-protocol COVE participants with replacement to generate realistic covariate distributions for $X=(X_f,X_{co},X_t)$ in the simulated datasets \cite{franklin_plasmode_2014, schreck_statistical_2024}. We retained participants' sex $X_f$ (48\% female) and presence of any comorbidities $X_{co}$ (23\% with comorbidities) as defined in the Supplementary Material and El Sahly et al. \cite{el_sahly_efficacy_2021}. $X_t$ represents calendar date of symptom-triggered SARS-CoV-2 testing, which we defined as the number of days after September 1, 2020 that participants either met the primary COVID-19 endpoint in COVE \cite{el_sahly_efficacy_2021} or were censored. Participants' actual vaccination and SARS-CoV-2 status during COVE were not used in simulations.

We generated indicators of having a high immune marker level $A$ and SARS-CoV-2 infection $Y$ using logistic regression models according to three confounding settings to compare method performance (Table \ref{tab:datagentmle} and Figure \ref{fig:aydatagen}). The main effects confounding setting, in which $A$ and $Y$ depend on $X_f,X_{co},$ and $X_t$ main effects, assesses if the TMLE approach has comparable performance to a standard TND analysis that is correctly specified \cite{bond_regression_2016}. The interaction confounding setting and splines confounding setting represent scenarios where the standard TND method is incorrectly specified and more flexible confounding adjustment via the TMLE may be beneficial. In the interaction confounding setting, comorbidity status modifies both the immune marker and sex association and the SARS-CoV-2 infection and sex association. In the splines confounding setting, the probability of a high immune marker level and probability of SARS-CoV-2 infection vary over time via calendar date splines. For each confounding setting, we generated the conditional association between $A$ and $Y$ as $\beta(P_F)=\log \text{OR}(P_F)(x)$ at $\log(0.2), \log(0.7),$ and $\log(1)$ for all subgroups in the healthcare-seeking population. 

Indicators of infection with a pathogen other than SARS-CoV-2 that could induce symptoms $W$ and meeting the symptom definition $D$ were generated using logistic regression models on the full healthcare-seeking population (Table \ref{tab:datagentmle}). To obtain the phase one participants in the TND study ($S=1$), we sampled without replacement $n=500, 1000, 2000,$ and $ 3000$ individuals who met the symptom definition ($D=1$) in the simulated healthcare-seeking population. Roughly $4-29\%$ of the phase one TND participants had COVID-19, depending on the confounding setting. 

Though we generated $A$ in the healthcare-seeking population, we only observed $A$ in TND participants with $\Delta=1$ (i.e., phase two TND cohort) using three two-phase sampling designs to assess how missing data patterns affect method performance (Table \ref{tab:datagentmle}). In the biased 1:1 case-noncase two-phase sampling design, we observed $A$ on all COVID-19 cases in the TND study cohort and an equal number of noncases. We sampled noncases without replacement across four covariate strata (females without comorbidities, males without comorbidities, females with comorbidities, and males with comorbidities), oversampling in strata with the lowest and highest probability of having a high immune marker level to improve efficiency \cite{gilbert_optimal_2014}. When TND study cohorts lacked enough noncases to meet the desired sampling ratios, we sampled individuals from additional strata to obtain enough noncases. In the biased 1:3 case-noncase two-phase sampling design, we observed $A$ on all cases and three times as many noncases, using the same biased covariate strata sampling. In the last design, we observed $ A$ on all cases and noncases in the TND study cohort. 

We compared seven estimators of $\beta(P)$, which identified $\beta(P_F)$ since the data-generating mechanism adhered to Section \ref{subsec:statparam} assumptions. We constructed the TMLE from Section \ref{sec:semi}, adjusting for covariates, $X_f,X_{co},$ and $X_t$, and allowing for two-way covariate interactions. We used the partially linear first-order smooth highly adaptive lasso for estimation \cite{benkeser_highly_2016}. We compared the TMLE to pseudo-likelihood logistic regression estimators, PLMx and PLEx, that adjust for covariate main effects and a sex-comorbidity interaction. PLMx and PLEx report identical point estimates, but PLMx reports model variance estimates, which perform better under correctly specified confounding relationships, and PLEx reports empirical variance estimates, which perform better under misspecified confounding relationships \cite{breslow_weighted_1997,haneuse_osdesign_2011}. These methods correctly account for the two-phase sampling designs using an offset and have been used in two-phase immune correlate studies \cite{breslow_logistic_1988,breslow_weighted_1997,haneuse_osdesign_2011, gilbert_powersample_2016,fong_modification_2018}. For comparison, we constructed two naïve pseudo-likelihood logistic regression estimators, nPLM and nPLE, that misspecify the two-phase sampling design by adjusting for covariate main effects. We also fit two maximum likelihood estimators of ordinary logistic regression models, MLEx and nMLE. MLEx adjusts for covariate main effects and a sex-comorbidity interaction and nMLE adjusts for covariate main effects. In the main effects confounding setting, all methods correctly specify the confounding distribution. In the interaction confounding setting, TMLE, PLMx, PLEx, and MLEx correctly specify the confounding distribution. In the splines confounding setting, only the TMLE is correctly specified. 

\subsection{Simulation results}
\label{subsec:simresults}

Across all simulation settings, the TMLE method had comparable or better bias, 95\% confidence interval (CI) coverage, type 1 error, and power than the comparison estimators (Figures \ref{fig:biasinxsplines}-\ref{fig:covinxsplines} and \ref{fig:biasmaineff}-\ref{fig:powertmle}). In the main effects confounding settings, the TMLE had similar levels of bias as all pseudo-likelihood logistic regression methods (PLMx, PLEx, nPLM, nPLE) and the MLEx across all three two-phase sampling designs, with less bias in larger phase one and phase two TND cohorts. All methods had similar 95\% CI coverage and type 1 error around the nominal level and power. In the interaction confounding setting, the TMLE was consistent, less biased, and nearly attained the nominal 0.95 level CI coverage and 0.05 type 1 error compared to the nPLM, nPLE, and nMLE. The TMLE had slightly worse bias and 95\% CI coverage compared to PLMx, PLEx, and MLEx, but performance improved with larger sample sizes. In nearly all splines confounding settings, the TMLE was consistent, had less bias, and nearly attained the nominal 0.95 level CI coverage and 0.05 type 1 error compared to the other estimators. 

The TMLE had a consistent variance estimator, was as variable, and had as small of variance estimates as the comparison estimators in most scenarios (Figure \ref{fig:varcomplete} and Table S2%\ref{tab:vartmle}
). The TMLE mean estimated standard error approximated the Monte Carlo standard deviation in most settings, but underestimated the Monte Carlo standard deviation in several 1:1 and 1:3 case-noncase two-phase sampling designs with fewer than two hundred phase two TND participants. The TMLE Monte Carlo standard deviations were roughly identical to the Monte Carlo standard deviations for PLMx, PLEx, and MLEx for all settings. In nearly all main effects and interaction confounding settings, the TMLE mean estimated standard error was similar to the mean estimated standard errors of PLMx, PLEx, and MLEx.

 \begin{figure}
     \centering
     \includegraphics[width=1\linewidth]{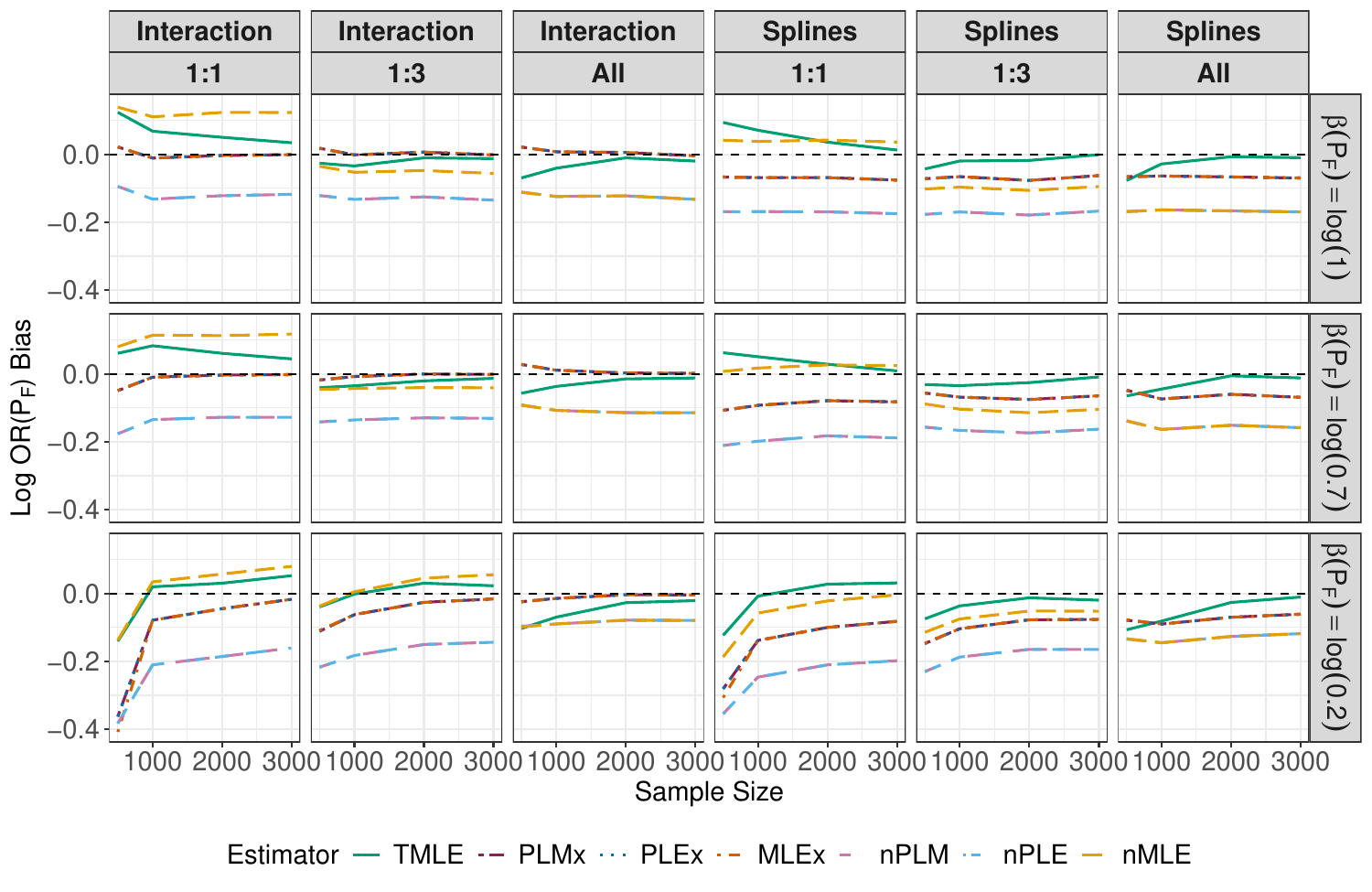}
     \caption{Bias of seven estimators from 1000 simulated two-phase test-negative design (TND) study cohorts generated under the interaction confounding setting and splines confounding setting. Phase one samples sizes were $500, 1000, 2000, $ and $3000$. Phase two TND participants ($A$ observed) were determined using a biased 1:1 case-noncase two-phase sampling design, biased 1:3 case-noncase two-phase sampling design, or all TND participants. The TMLE adjusts for covariates and uses highly adaptive lasso for estimation. PLMx, PLEx, and MLEx denote two pseudo-likelihood logistic regression approaches (model variance or empirical variance) and an ordinary logistic regression, respectively, that adjust for covariate main effects and an interaction. nPLM, nPLE, and nMLE denote two naïve pseudo-likelihood logistic regression approaches and a naïve ordinary logistic regression that adjust for covariate main effects.}
     \label{fig:biasinxsplines}
 \end{figure}

 \begin{figure}
     \centering
     \includegraphics[width=1\linewidth]{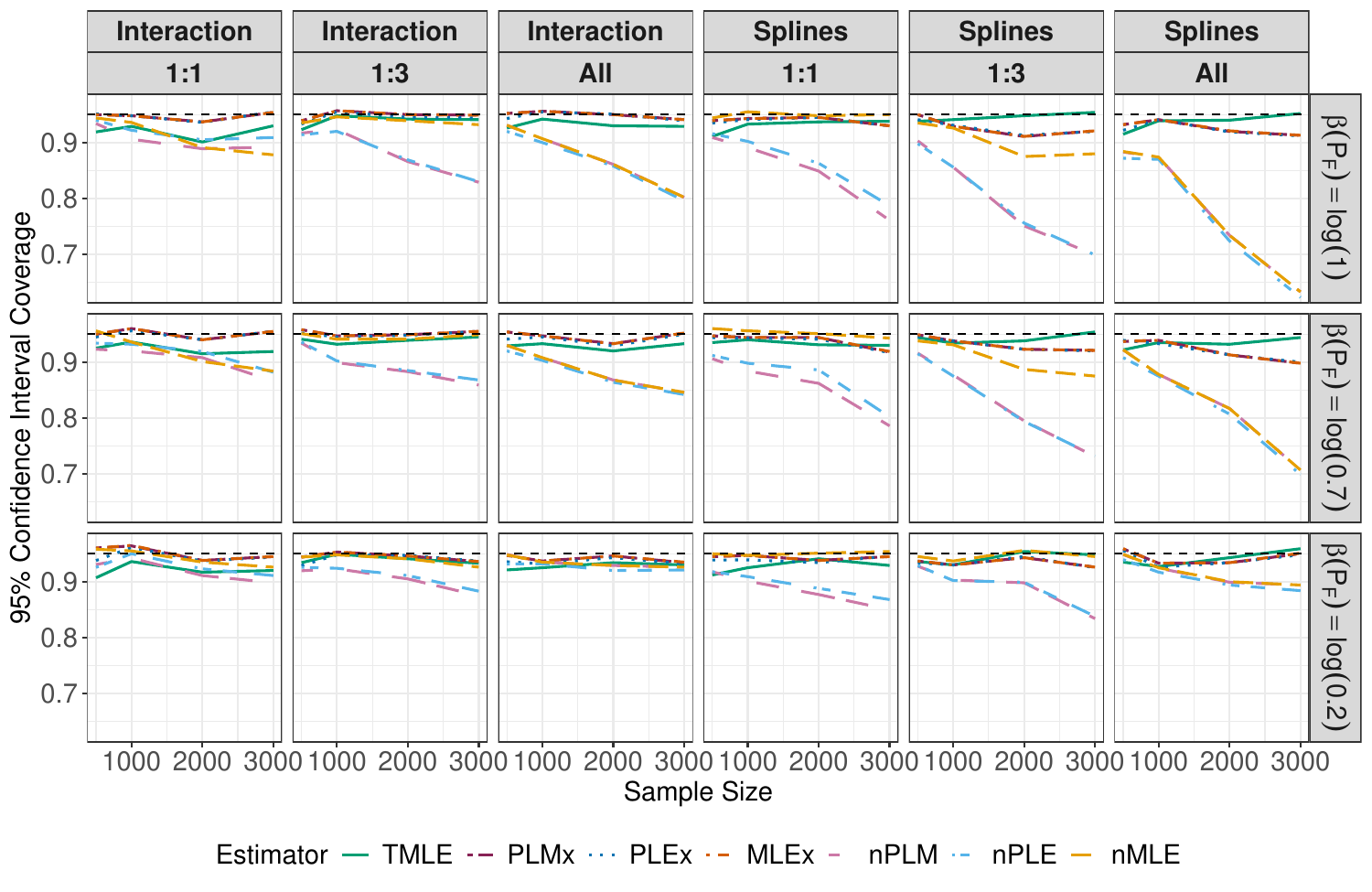}
     \caption{95\% confidence interval coverage of seven estimators from 1000 simulated two-phase test-negative design (TND) study cohorts generated under the interaction confounding setting and splines confounding setting. Phase one samples sizes were $500, 1000, 2000, $ and $3000$. Phase two TND participants ($A$ observed) were determined using a biased 1:1 case-noncase two-phase sampling design, biased 1:3 case-noncase two-phase sampling design, or all TND participants. The TMLE adjusts for covariates and uses highly adaptive lasso for estimation. PLMx, PLEx, and MLEx denote two pseudo-likelihood logistic regression approaches (model variance or empirical variance) and an ordinary logistic regression, respectively, that adjust for covariate main effects and an interaction. nPLM, nPLE, and nMLE denote two naïve pseudo-likelihood logistic regression approaches and a naïve ordinary logistic regression that adjust for covariate main effects.}
     \label{fig:covinxsplines}
 \end{figure}

\section{Moderna COVE trial data applications}
\label{sec:app}

We applied our method to assess COVID-19 vaccine effectiveness and several potential exposure-proximal immune correlates of COVID-19 using TND study cohorts obtained from COVE during the blinded phase from July 2020 to March 2021. 

The COVE (mRNA-1273-P301) study was conducted in accordance with the International Council for Harmonisation of Technical Requirements for Pharmaceuticals for Human Use, Good Clinical Practice guidelines, and applicable government regulations. The study was also conducted in compliance with the ethical principles that have their origin in the Declaration of Helsinki. The Central Institutional Review Board approved the mRNA-1273-P301 protocol and the consent forms. The mRNA-1273-P301 protocol (Pro00044270) was initially approved on June 18, 2020. Central Institutional Review Board services for the mRNA-1273-P301 study were provided by Advarra, Inc., 6100 Merriweather Dr., Suite 600, Columbia, MD 21044. All necessary patient/participant informed consent before enrollment has been obtained and the appropriate institutional forms have been archived. 

\subsection{Application 1: inference on COVID-19 vaccine effectiveness}

\label{subsec:appve}

We constructed a TND study cohort from 28,451 per-protocol COVE participants to assess vaccine effectiveness against the primary COVID-19 endpoint from COVE, or primary COVID-19 
\cite{el_sahly_efficacy_2021}. Per-protocol participants received two doses of the intervention they were blinded and randomized to (mRNA-1273 vaccine or placebo), had no major protocol violations during the blinded phase, and were naïve to SARS-CoV-2. Vaccination status was measured on all participants. We identified 2,553 participants who met the primary COVID-19 symptom definition and had at least one positive or negative SARS-CoV-2 test at least two weeks after the second intervention dose, within ten days after symptom onset, after meeting the symptom definition, while blinded, and before receiving any nonstudy COVID-19 vaccinations (Table \ref{tab:vetmletbl1}). We applied participant-based sampling with censoring for COVID-19 to select one SARS-CoV-2 test per person \cite{de_serres_test-negative_2013, andrews_evaluating_2025}: 728 participants were classified as cases and 1,825 participants were classified as noncases.

To demonstrate a post-marketing TND analysis that must address confounding, we adjusted for age, sex, race and ethnicity (Person of Color vs. Non-Hispanic/Latino White), presence of comorbidities, US Census region (Midwest, Northeast, South, and West), and quantitative or two-week calendar date of SARS-CoV-2 testing, depending on the statistical method used \cite{el_sahly_efficacy_2021,us_census_bureau_geographic_2021,lopez_bernal_effectiveness_2021, chua_use_2020,belongia_effectiveness_2009, bond_regression_2016}.  See the Supplementary Material for additional details. 

To determine if we can identify the conditional risk ratios of COVID-19 in the healthcare-seeking population, we evaluated identifying Assumptions \ref{ass:posy}-\ref{ass:ymisclass}. Assumption \ref{ass:posy} is satisfied because rhinoviruses, respiratory enteroviruses, and adenoviruses caused infections at their typical seasonal levels during the pandemic, ensuring all subgroups had some chance of experiencing COVID-19 symptoms from SARS-CoV-2 or other causes \cite{olsen_changes_2021, chow_effects_2023, ashby_validating_2025}. Assumption \ref{ass:poss} is satisfied because individuals in all four US Census regions had access to SARS-CoV-2 nucleic acid amplification testing \cite{centers_for_disease_control_and_prevention_covid_2020}. Assumption \ref{ass:condsa} is reasonable because COVE randomized and blinded vaccination status. In post-marketing TND studies, Assumption \ref{ass:condsa} holds because we limit inference to the healthcare-seeking population and assume that vaccination status is only associated with SARS-CoV-2 testing through healthcare-seeking behavior or SARS-CoV-2 infection and meeting the symptom definition (Figure \ref{fig:dag1}). Noncase Exchangeability is reasonable because vaccine-induced immune responses are antigen-specific \cite{male_immunology_2013, huang_non-hiv_2024}. Andrews et al. \cite{andrews_evaluating_2025} and Ashby et al. \cite{ashby_validating_2025} also validated that the mRNA-1273 vaccine does not affect other causes of COVID-19-like symptoms. To reduce case status misclassification, our TND study cohort obtained SARS-CoV-2 tests within 10 days of symptom onset \cite{kucirka_variation_2020,patel_evaluation_2021, andrews_evaluating_2025}. Assumption \ref{ass:posdelta}, \ref{ass:conddeltaa}, and \ref{ass:amisclass} are satisfied because vaccination status was accurately observed for all per-protocol COVE participants. Post-marketing TND studies should adjust for covariates related to vaccination status missingness, such as region and age. Compared to other observational studies that rely on self-report or linking vaccine records, TND studies are less prone to bias from vaccination status misclassification since vaccination status is collected before SARS-CoV-2 status is known \cite{jackson_use_2019,liang_typhoid_2023,rolnick_self-report_2013,zimmerman_sensitivity_2003, sullivan_theoretical_2016}.

We also evaluated Assumptions \ref{ass:cons}-\ref{ass:confound} to determine if we can identify the causal conditional risk ratio and interpret both estimators as vaccine effectiveness, $VE = (1-RR(P_{F,ca}))\times 100\%=(1-OR(P))\times 100\%$. Consistency is satisfied because we do not expect any variability within vaccination statuses. No interference is reasonable because vaccines were not readily available in the US during the blinded phase of COVE and participants were from 99 study sites and unlikely to interact with each other \cite{el_sahly_efficacy_2021}. There was no unmeasured confounding in the TND cohort. Thus, it is appropriate to estimate conditional odds ratio estimators using the observed TND data and interpret them as vaccine effectiveness estimators.

We applied two statistical methods to estimate primary COVID-19 vaccine effectiveness in the TND study cohort. We applied our proposed TMLE method adjusted for the aforementioned covariates and used ensemble super learner for nuisance estimation \cite{van_der_laan_super_2007, chernozhukov_doubledebiased_2018} (Table \ref{tab:vesl}). We also fit an ordinary logistic regression adjusted for covariate linear main effects \cite{lopez_bernal_effectiveness_2021, chua_use_2020, belongia_effectiveness_2009, bond_regression_2016}. 

To assess how the TND and RCT study design and analysis methods compare in a setting without confounding or selection bias, we also reported the COVID-19 vaccine efficacy estimates and 95\% CIs from the COVE final blinded phase analysis publication \cite{el_sahly_efficacy_2021}. RCT vaccine efficacy was defined as one minus the COVID-19 hazard ratio (vaccine vs. placebo), estimated using a Cox proportional hazards model stratified on randomization factors and Efron’s method for handling ties \cite{el_sahly_efficacy_2021,efron_efficiency_1977}. 

The TND vaccine effectiveness and RCT vaccine efficacy estimates were highly concordant (Figure \ref{fig:coveve}). Using the TMLE approach, the mRNA-1273 vaccine effectiveness against symptomatic COVID-19 was 92.8\% (95\% CI = 90.1 to 94.8) in the healthcare-seeking population, holding covariates constant. The TMLE and ordinary logistic regression approaches produced similar vaccine effectiveness estimates and 95\% CIs that overlapped with the published RCT vaccine efficacy estimates (vaccine efficacy = 93.2, 95\% CI = 91.0 to 94.8). Despite the TND study cohorts being 9\% the size of the RCT cohorts, the TMLE variance estimates were only 33\% larger than the RCT variance estimates on the natural logarithm scale. 

This data application illustrates that when confounding is controlled and all TND participants have the same healthcare-seeking behavior, TND vaccine effectiveness estimates approximate RCT vaccine efficacy estimates.

\begin{figure}
    \centering
    \includegraphics[width=1\linewidth]{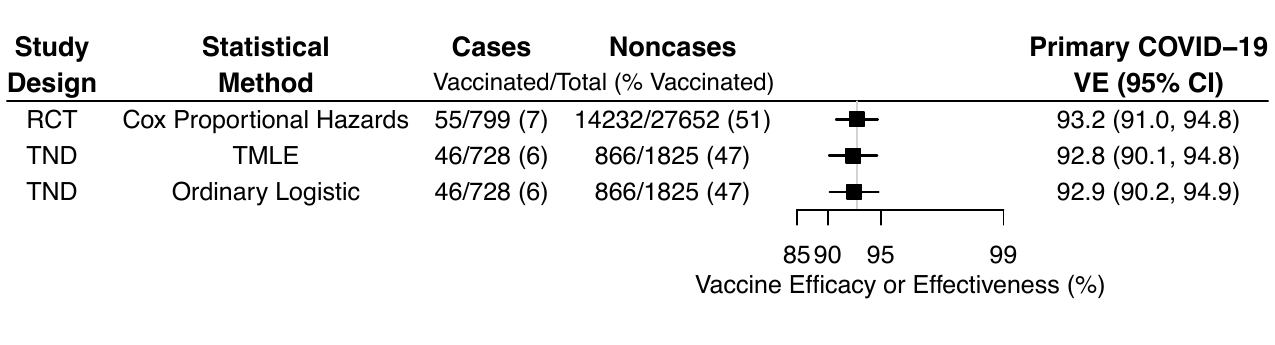}
    \caption{Moderna COVE primary COVID-19 vaccine efficacy and vaccine effectiveness estimates. Primary COVID-19 RCT vaccine efficacy was defined as one minus the COVID-19 hazard ratio, estimated on on the per-protocol COVE RCT cohort using a stratified Cox proportional hazards model and the Efron method for handling ties \cite{el_sahly_efficacy_2021, efron_efficiency_1977}. Primary COVID-19 TND vaccine effectiveness was defined as one minus the causal conditional COVID-19 risk ratio, estimated on the COVE TND study cohort. The TMLE approach adjusted for age, sex, race and ethnicity, comorbidities, region, and calendar date and used ensemble super learner for estimation. The ordinary logistic regression adjusted for covariate main effects.
Estimates and 95\% CIs are compared on the ln(1-VE) scale, with plotting labels on the VE scale.    
\\
  Abbreviations: RCT = Randomized Placebo-Controlled Clinical Trial; TND = Test-Negative Design; VE = Vaccine Effectiveness or Vaccine Efficacy; CI = Confidence Interval; ln = natural logarithm}
    \label{fig:coveve}
\end{figure}

\subsection{Application 2: inference on immune markers as exposure-proximal correlates of COVID-19}

\label{subsec:appimm}

To illustrate how the TND and our TMLE approach can be used to evaluate immune correlates, we reanalyzed COVE data \cite{el_sahly_efficacy_2021} as a two-phase TND immune correlates study. We evaluated if neutralizing antibodies (nAbs) and binding antibodies (bAbs) are correlates of COVID-19 in healthcare-seeking individuals who are naïve to SARS-CoV-2 and received two mRNA-1273 doses within six months. We studied a COVID-19 symptom definition recommended by the Centers for Disease Control and Prevention (CDC) in 2020 \cite{cdc_symptoms_2022, stokes_coronavirus_2020}. 

Our phase one TND study cohort consisted of 935 per-protocol COVE participants who received both mRNA-1273 doses, met the CDC COVID-19 symptom definition, and had at least one positive or negative SARS-CoV-2 test at least 64 days after their first dose, within ten days after symptom onset, after meeting the symptom definition, while blinded, and before receiving any nonstudy COVID-19 vaccinations (Table \ref{tab:immtmletbl1}). The 64-day waiting period was an artifact of the original data collection \cite{gilbert_immune_2022,follmann_test-negative_2025}; SARS-CoV-2 tests obtained once individuals have a well-defined vaccination status at potential SARS-CoV-2 exposure is sufficient for post-marketing TND immune correlates studies. Like Section \ref{subsec:appve}, we applied participant-based sampling with censoring for COVID-19 \cite{de_serres_test-negative_2013,andrews_evaluating_2025} to obtain 46 cases and 889 noncases. 

During COVE, symptomatic participants had their blood drawn on Disease Day 1, or symptom onset \
\cite{el_sahly_efficacy_2021, gilbert_immune_2022}, and stored for immunological assessment. Our phase two TND study cohort consisted of 127 of these participants (34 cases and 93 noncases) who were anti-nucleocapsid seronegative and had complete 50\% inhibitory dilution nAb titers and immunoglobulin G (IgG) bAb concentrations measured against the ancestral SARS-CoV-2 spike protein and receptor-binding domain (RBD) from previous COVE immune correlate studies (Figure \ref{fig:immcorr}) \cite{gilbert_immune_2022, follmann_test-negative_2025}. While the original immune correlates data was collected via 1:3 case-noncase matching by US Census region and testing date \cite{follmann_test-negative_2025}, we analyzed these participants as if they were sampled without replacement from our phase one TND cohort to attain a 1:3 case-noncase ratio within each stratum defined by US Census region \cite{us_census_bureau_geographic_2021}.

We investigated nAb titers, anti-spike IgG bAb concentration, and anti-RBD IgG bAb concentration dichotomized by a threshold to study if some thresholds are associated with a greater demarcation in COVID-19 risk. The limits of detection for nAb titers, anti-spike IgG bAb concentration, and anti-RBD IgG bAb concentration are not suitable thresholds because few individuals had immune marker levels under these limits given the mRNA-1273 vaccine's high efficacy in SARS-CoV-2 naïve individuals \cite{el_sahly_efficacy_2021, gilbert_immune_2022, follmann_test-negative_2025}. Instead, we investigated thresholds (greater than or equal to vs. below) defined from the 20th, 30th, 40th, 50th, 60th, 70th, and 80th percentile for each immune marker. 

We adjusted for risk score of acquiring COVID-19, presence of comorbidities, US Census region, and quantitative or tertile calendar date of SARS-CoV-2 testing, depending on the statistical method used \cite{gilbert_immune_2022,follmann_test-negative_2025, el_sahly_efficacy_2021,us_census_bureau_geographic_2021}. The risk score of acquiring COVID-19 measured the elevated risk for COVID-19 and was estimated from a super learner analysis of the placebo arm of COVE in Gilbert et al. \cite{gilbert_immune_2022}. 

We evaluated identifying Assumptions \ref{ass:posy}-\ref{ass:ymisclass} to determine if we can identify the conditional risk ratios of CDC COVID-19 in a SARS-CoV-2 naïve healthcare-seeking population that has been vaccinated with mRNA-1273 within six months. As discussed in Section \ref{subsec:appve}, Assumptions \ref{ass:posy}, \ref{ass:poss}, and \ref{ass:ymisclass} are satisfied because they do not depend on the exposure variable. Assumption \ref{ass:posdelta} is satisfied because blood samples were collected on all participants with symptoms during COVE, so all covariate subgroups and case statuses had some probability of having their immune markers measured \cite{gilbert_immune_2022, follmann_test-negative_2025}. To satisfy Assumption \ref{ass:condsa}, we adjusted for variables that may be associated with immune marker levels and healthcare-seeking behavior, like comorbidities, risk score of acquiring COVID-19, region, and calendar date. For Assumption \ref{ass:conddeltaa}, we adjusted for region because it was used in the two-phase sampling design \cite{follmann_test-negative_2025}. Noncase Exchangeability is reasonable because the studied nAbs and bAbs are unlikely to be associated with other causes of symptoms, after adjusting for comorbidities and risk score of acquiring COVID-19, which could affect the immune system. Similarly, Follmann et al. \cite{follmann_test-negative_2025} found that the nAb titers and anti-spike IgG bAb concentrations measured in symptomatic SARS-CoV-2 negative individuals at testing were similar to the predicted titers and concentrations of individuals who were not symptomatic and did not have COVID-19 on the same date. After SARS-CoV-2 infection, SARS-CoV-2 nAbs and anti-RBD IgG bAbs emerge around the second week after symptom onset \cite{painter_prior_2023,follmann_kinetics_2023}. Post-marketing TND nAb and bAb studies should limit enrollment to individuals who obtain SARS-CoV-2 testing within five or seven days after symptom onset to avoid measuring nAb titers and bAb concentrations perturbed from recent SARS-CoV-2 infection \cite{sumner_antisars-cov-2_2024}. Since our TND study included SARS-CoV-2 tests within ten days after symptom onset, we may underestimate the benefit of high immune marker levels. We also restricted the phase one TND study cohort to anti-nucleocapsid seronegative SARS-CoV-2 naïve individuals because Follmann et al. \cite{follmann_test-negative_2025} reported that nAb titers and anti-spike IgG bAb concentrations measured at testing approximate immune marker levels present at the time of potential SARS-CoV-2 exposure.

We also evaluated Assumptions \ref{ass:cons}-\ref{ass:confound} to determine if we can identify the causal conditional risk ratio of CDC COVID-19. Since our population of inference is healthcare-seeking individuals who are naïve to SARS-CoV-2 and received two mRNA-1273 doses within six months, all immune marker levels are attainable, though consistency may be violated if COVID-19 risk varies within immune marker levels. No interference is reasonable since participants were from many sites and unlikely to interact with each other. To account for Assumption \ref{ass:confound}, we adjusted for risk score of acquiring COVID-19, comorbidities, region, and calendar date. Overall, the causal identification is appropriate.  

We applied three statistical methods to estimate the causal risk ratio of CDC COVID-19 when comparing healthcare-seeking individuals with nAb titers, anti-spike IgG bAb concentrations, or anti-RBD IgG bAb concentrations greater than or equal to vs. below a given threshold. For each immune marker, we calculated the TMLE using the phase two TND study cohort, adjusting for the aforementioned covariates and applying ensemble super learner for estimation (Table \ref{tab:immsl}). For our second statistical method, we constructed a pseudo-likelihood logistic regression estimator with empirical standard error estimates, which adjusted for covariate linear main effects and accounted for region in the two-phase sampling design \cite{breslow_logistic_1988, breslow_weighted_1997}. Lastly, we fit an ordinary logistic regression adjusted for covariate linear main effects using the phase two TND study cohort. 

We found no evidence that a specific threshold of nAb titer discriminates CDC COVID-19 risk (Figure \ref{fig:neut}). The CDC COVID-19 risk ratio for each nAb titer threshold is around 0.65 with 95\% CIs including 1 for all three statistical methods. Since all the CDC COVID-19 risk ratio estimates were similar and less than one, this could suggest a continuous relationship between detectable nAb titer and CDC COVID-19 risk. 

For bAbs, it appears that low thresholds do not discriminate CDC COVID-19 risk, but bAb concentrations compared at higher thresholds provide significantly different risks of CDC COVID-19 (Figure \ref{fig:spike} and \ref{fig:rbd}). For example, there is no evidence that anti-RBD IgG bAb concentrations have significantly different CDC COVID-19 risks at thresholds below 1,945 binding antibody units per milliliter (BAU/ml). However, the estimated risk of CDC COVID-19 using the TMLE is 78\% lower (95\% CI: 34\% to 93\% lower) for individuals with at least 2,576 anti-RBD IgG bAb BAU/ml compared to individuals below that threshold who share the same covariates. We found similar trends with anti-spike IgG bAb concentration and CDC COVID-19 though they did not reach statistical significance.

While the TMLE, pseudo-likelihood logistic regression, and ordinary logistic regression produced similar results across immune markers and thresholds, the TMLE estimates were slightly higher than the other estimates. This may suggest that the simpler analysis methods slightly underestimated the causal conditional risk ratios.

Previous immune correlates studies have shown that nAbs, anti-spike IgG bAbs, and anti-RBD IgG bAbs induced by the mRNA-1273 vaccine are inversely correlated with risk of COVID-19 \cite{khoury_neutralizing_2021,corbett_immune_2021,gilbert_immune_2022, follmann_test-negative_2025}. Our nAb analyses were consistent with previous literature, but our anti-spike and anti-RBD IgG bAb analyses suggested that this inverse relationship exists only at higher thresholds. Given the small phase two TND study cohort, our analysis was exploratory; future TND immune correlates studies should measure immune markers from more participants to improve power and precision. 

\begin{figure}
    \centering
    \includegraphics[width=1\linewidth]{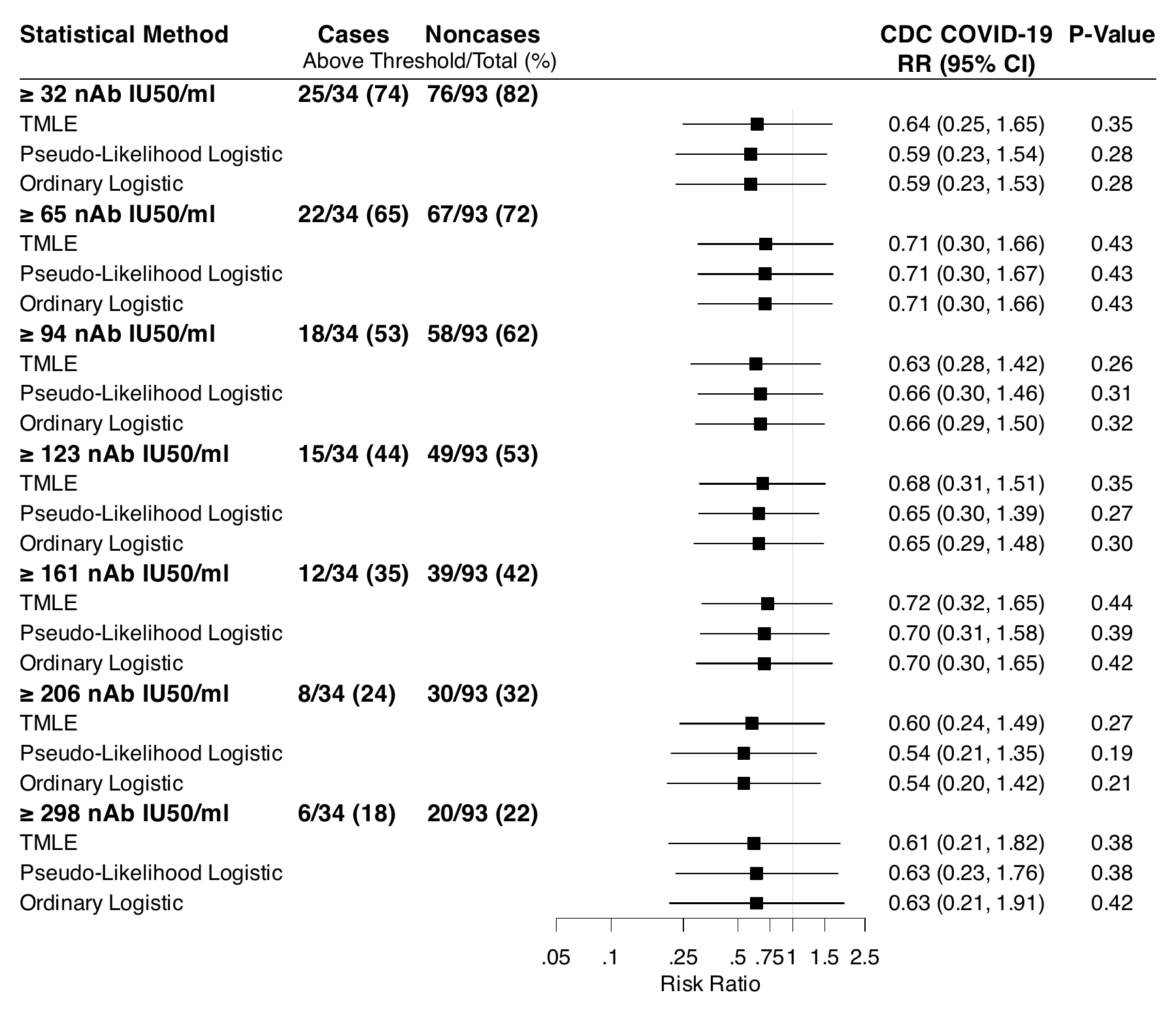}
    \caption{Causal conditional risk ratios of CDC COVID-19 by 50\% inhibitory dilution neutralizing antibody titers dichotomized by a threshold from the Moderna COVE two-phase test-negative design immune correlates study. Risk ratios compare subgroups greater than or equal to vs. below the 20th, 30th, 40th, 50th, 60th, 70th, and 80th percentile marker value. The TMLE approach adjusted for risk score of acquiring COVID-19, comorbidities, region, and calendar date and used an ensemble super learner for estimation. The pseudo-likelihood logistic regression with empirical standard error estimates accounted for the two-phase sampling design and adjusted for covariate main effects. The ordinary logistic regression adjusted for covariate main effects. Estimates are on the natural logarithm risk ratio scale, with plotting labels on the risk ratio scale.   
\\
  Abbreviations: CDC = Centers for Disease Control and Prevention; RR = Risk Ratio; CI = Confidence Interval; nAb = Neutralizing Antibody; IU50 = 50\% Inhibitory Units; ml = Milliliter}
    \label{fig:neut}
\end{figure}

\begin{figure}
    \centering
    \includegraphics[width=1\linewidth]{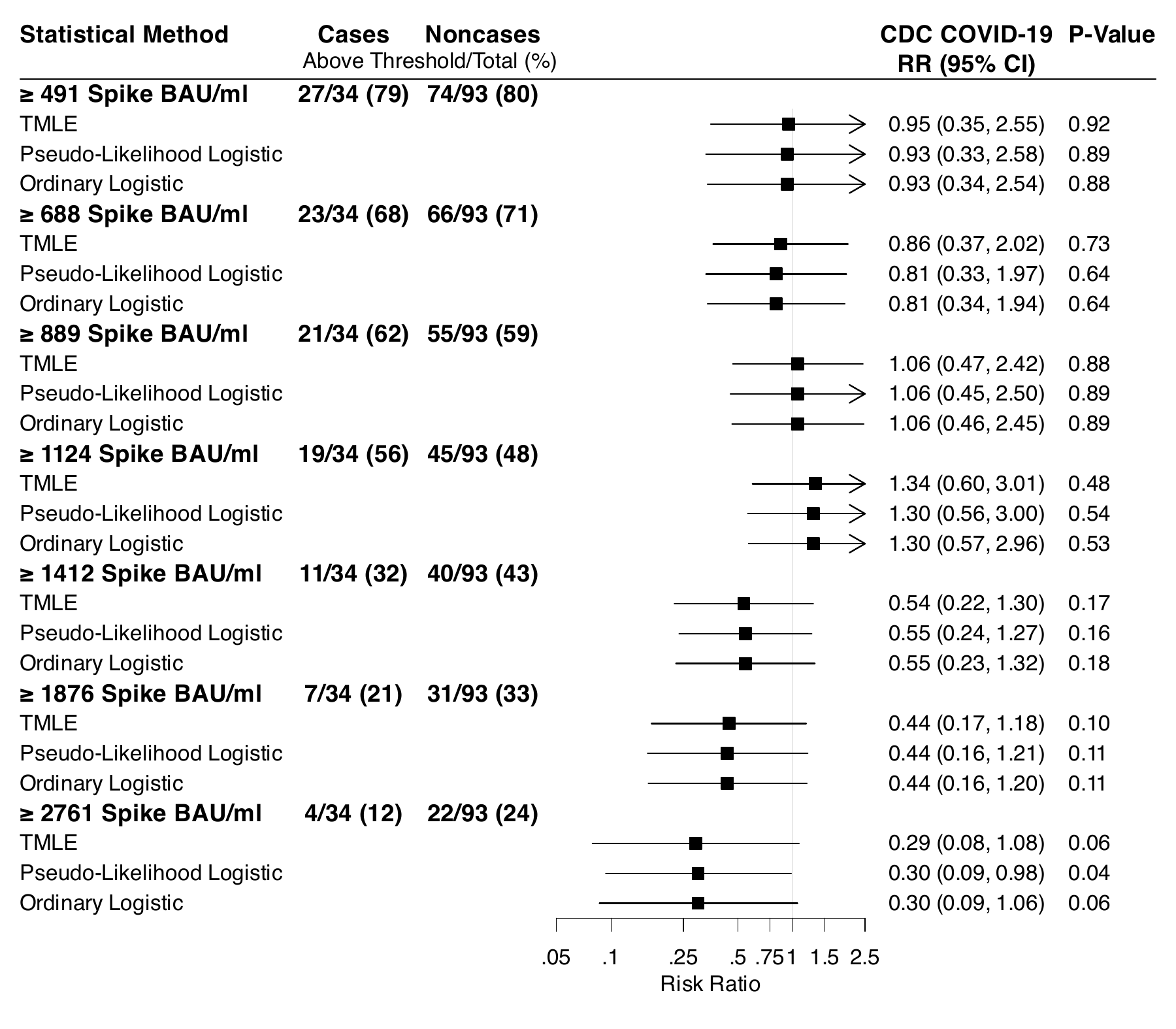}
    \caption{Causal conditional risk ratios of CDC COVID-19 by anti-spike immunoglobulin G binding antibody concentration dichotomized by a threshold 
from the Moderna COVE two-phase test-negative design immune correlates study. Risk ratios compare subgroups greater than or equal to vs. below the 20th, 30th, 40th, 50th, 60th, 70th, and 80th percentile marker value. The TMLE approach adjusted for risk score of acquiring COVID-19, comorbidities, region, and calendar date and used an ensemble super learner for estimation. The pseudo-likelihood logistic regression with empirical standard error estimates accounted for the two-phase sampling design and adjusted for covariate main effects. The ordinary logistic regression adjusted for covariate main effects. Estimates are on the natural logarithm risk ratio scale, with plotting labels on the risk ratio scale.  
\\
  Abbreviations: CDC = Centers for Disease Control and Prevention; RR = Risk Ratio; CI = Confidence Interval; BAU = Binding Antibody Units; ml = Milliliters}
    \label{fig:spike}
\end{figure}

\begin{figure}
    \centering
    \includegraphics[width=1\linewidth]{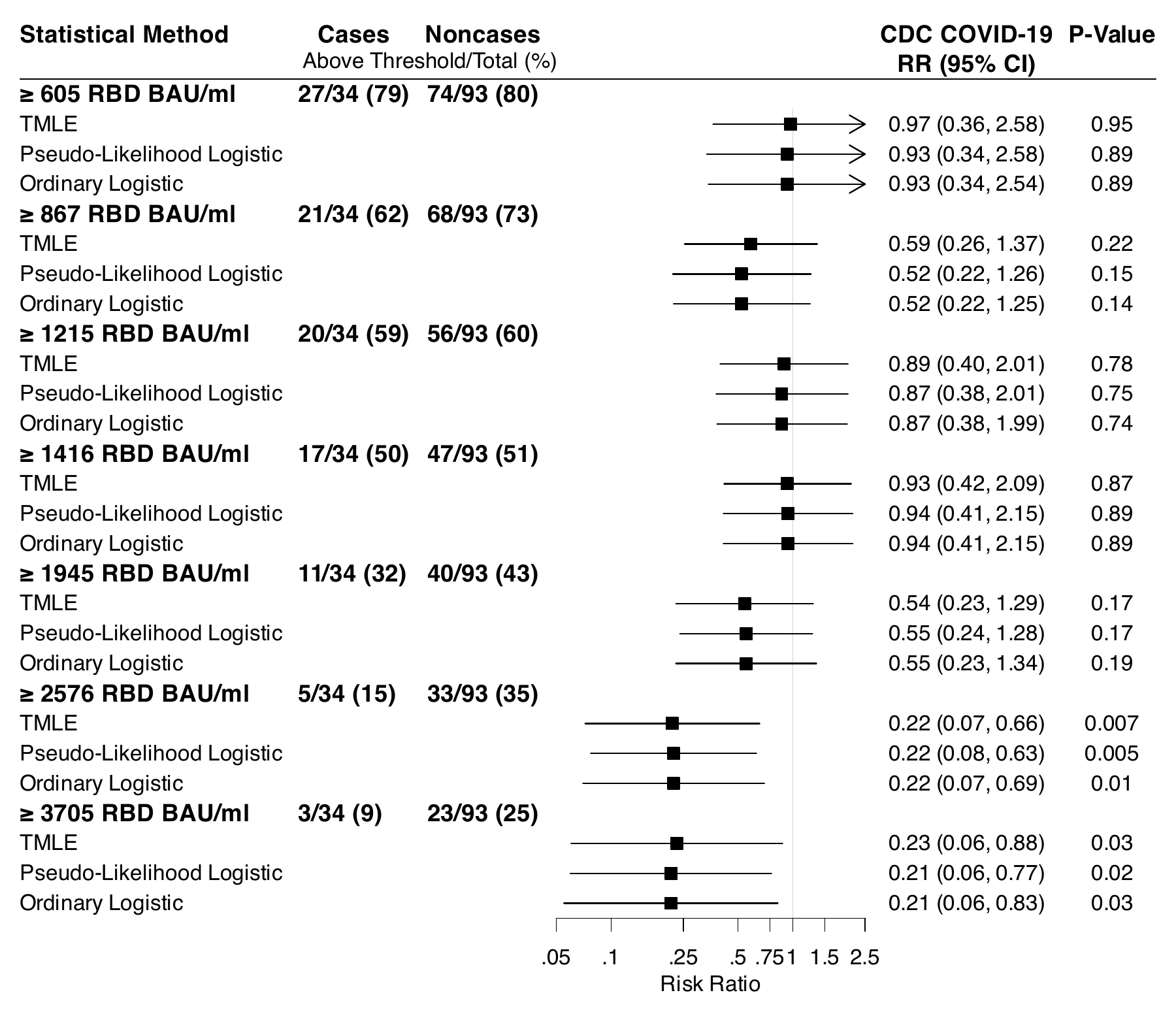}
    \caption{Causal conditional risk ratios of CDC COVID-19 by anti-receptor-binding domain immunoglobulin G binding antibody concentration dichotomized by a threshold from the Moderna COVE two-phase test-negative design immune correlates study. Risk ratios compare subgroups greater than or equal to vs. below the 20th-80th percentile marker value. The TMLE approach adjusted for risk score of acquiring COVID-19, comorbidities, region, and calendar date and used an ensemble super learner for estimation. The pseudo-likelihood logistic regression with empirical standard error estimates accounted for the two-phase sampling design and adjusted for covariate main effects. The ordinary logistic regression adjusted for covariate main effects. Estimates are on the natural logarithm risk ratio scale, with plotting labels on the risk ratio scale.    
\\
  Abbreviations: CDC = Centers for Disease Control and Prevention; RR = Risk Ratio; CI = Confidence Interval; RBD = Receptor-Binding Domain; BAU = Binding Antibody Unit; ml = Milliliter}
    \label{fig:rbd}
\end{figure}

\section{Discussion}
\label{sec:discuss}

We extended the targeted maximum likelihood estimation methodology developed by van der Laan and Gilbert \cite{van_der_laan_semiparametric_2025} to a TND setting to study the effect of vaccine regimens and immune markers on symptomatic disease when confounding and missing exposure variable data are present. We presented identifying assumptions to obtain a causal conditional risk ratio and employed a semiparametric approach involving machine-learning for flexible confounding control. The resulting estimator is efficient and asymptotically linear when the exposure variable is missing at random and has favorable finite sample properties. We also demonstrated how our TMLE can assess COVID-19 vaccine effectiveness and antibody marker correlates of COVID-19 from a TND study using data from COVE.  

TND studies have been essential for monitoring post-marketing vaccine effectiveness, but are subject to bias from confounding and missing data. With fewer required assumptions than standard analysis methods, our TMLE method accounts for confounding using machine-learning methods, produces similar variance estimates, respects model constraints, and provides valid inference \cite{van_der_laan_targeted_2006, van_der_laan_targeted_2011}. The TMLE excels in high-dimensional large sample settings, like TND vaccine effectiveness studies \cite{van_der_laan_semiparametric_2025,chang_effectiveness_2022}, though machine-learning methods can be adapted to accommodate the amount of data available  \cite{van_der_laan_unified_2003-1,dudoit_asymptotics_2005,van_der_vaart_oracle_2006, van_der_laan_super_2007}. Smaller TND immune correlates studies may consider applying super learner with less flexible machine-learning models, log-likelihood loss with leave-one-out cross-validation, and/or learners with screens that limit the number of covariates in each model to improve small sample performance \cite{phillips_practical_2023}. To our knowledge, this is the first application of TMLE and second application of machine-learning algorithms for data-adaptive confounding control in TND analyses \cite{jiang_double_2025}.

Currently, our TMLE approach is restricted to accommodating binary exposure variables $A$. Since our method requires regressing $A$ using a partially linear logistic regression, we would have to assume a different model for a categorical or quantitative exposure variable and adapt the identifying assumptions. While there are heterogeneous vaccination statuses, comparing receipt of the current vaccine vs. no vaccination within a time period is useful for monitoring vaccine effectiveness. Similarly, while immune markers are often quantitative measurements, relevant binary thresholds exist and can be informative. Regulators have used fixed thresholds accepted to correspond to very low disease risk for several licensed vaccines \cite{plotkin_correlates_2010, chen_threshold_2013}. Immune markers' limit of detection can also be relevant because variability below the limit likely reflects technical measurement error that is uncorrelated with risk \cite{moodie_neutralizing_2018}. Additionally, multiple immune markers may be considered in a threshold, such as comparing individuals with high nAb titers and high memory B cell levels to high nAb titers and low memory B cell levels. A threshold approach is simpler than modeling multiple quantitative immune markers together. Future work may include extending our TMLE method for quantitative exposure variables. 

Another limitation is that our method assumes healthcare-seeking behavior is binary and restricts inference to the healthcare-seeking population. This population can be difficult to define and excludes individuals with limited access to medical care. Additionally, residual bias may exist if healthcare-seeking behavior is not binary and/or TND participants have varying healthcare-seeking behavior. Thus, our approach reduces rather than eliminates bias from healthcare-seeking behavior. Li et al. \cite{li_double_2023} and Li et al. \cite{li_doubly_2022} have developed causal methods that leverage negative control variables to account for confounding and bias from healthcare-seeking behavior and generalize to the entire population. Extending our TMLE to incorporate negative control variables could help reduce residual bias and produce more generalizable estimates. 

Given its practicality and cost-effectiveness, the TND has become a prominent study design for evaluating vaccines and a novel study design for investigating immune correlates of disease. By updating the statistical methods implemented, we can provide more accurate and interpretable estimates and conduct inference for a range of circumstances.

\section*{Acknowledgements}
The authors would like to thank the COVID-19 Prevention Network and the Moderna COVE study participants and study team for their contributions. 

\FloatBarrier

% ========== Appendix A

\section*{Appendix}
\label{sec:appendix}
\setcounter{subsection}{0}
\renewcommand{\thesubsection}{A.\arabic{subsection}}

\subsection{Theorem \ref{thm:fulldatarr} proof}

\begin{proof}
   First, we show that the observed TND data conditional odds ratio $OR(P)(x)$ identifies the full data conditional odds ratio $OR(P_F)$. 
\begin{align*}
&OR(P)(x)\\
&\overset{\textcolor{white}{A2.4}}{=}\frac{\mu_P(1,x)/(1-\mu_P(1,x))}{\mu_P(0,x)/(1-\mu_P(0,x))}\\
&\overset{\textcolor{white}{A2.4}}{=}\frac{P( A=1|\Delta = 1, S=1, Y=1, D=1, X=x)/P(A=0|\Delta = 1,S=1, Y=1, D=1, X=x)}{P( A=1|\Delta = 1,S=1, Y=0, D=1, X=x)/P( A=0|\Delta = 1,S=1,Y=0, D=1, X=x)}\\
&\overset{A\ref{ass:conddeltaa}}{=}\frac{P( A=1| S=1, Y=1, D=1, X=x)/P(A=0|S=1, Y=1, D=1, X=x)}{P( A=1|S=1, Y=0, D=1, X=x)/P( A=0|S=1,Y=0, D=1, X=x)} \\
&\overset{A\ref{ass:condsa}}{=}\frac{P( A=1|Y=1, D=1, X=x)/P(A=0| Y=1, D=1, X=x)}{P( A=1| Y=0, D=1, X=x)/P( A=0|Y=0, D=1, X=x)} \\
&\overset{\textcolor{white}{A2.4}}{=}\frac{P(Y=1, D=1|A=1,X=x)/P(Y=0,D=1|A=1, X=x)}{P(Y=1,D=1|A=0, X=x)/P(Y=0,D=1|A=0, X=x)}\\
&\underset{A\ref{ass:ymisclass}}{\overset{A\ref{ass:amisclass}}{=}}\frac{P_F(Y=1, D=1|A=1,X=x)/P_F(Y=0,D=1|A=1, X=x)}{P_F(Y=1,D=1|A=0, X=x)/P_F(Y=0,D=1|A=0, X=x)}\\
&\overset{\textcolor{white}{A2.4}}{=}OR(P_F)(x)
\end{align*}
The first and second lines are by definition using Equations \eqref{eqn:or} and \eqref{eqn:mu}. The third line follows from Assumption \ref{ass:conddeltaa}, the fourth line follows from Assumption \ref{ass:condsa}, the fifth line follows from Bayes rule, and the sixth line follow from Assumptions \ref{ass:amisclass} and \ref{ass:ymisclass}. Assumptions \ref{ass:posy}, \ref{ass:posdelta}, and \ref{ass:poss} are necessary to ensure all the conditional probabilities are well-defined.

$OR(P_F)(x)$ is a ratio of two odds estimands that are difficult to interpret. The odds in the numerator compare the probability of COVID-19 (i.e., SARS-CoV-2 infection and meeting the symptom definition) to the probability of meeting the symptom definition from a cause other than SARS-CoV-2 for individuals with $A=1$ and the same characteristics $X=x$. The odds in the denominator compare the probability of COVID-19 to the probability of meeting the symptom definition from a cause other than SARS-CoV-2 for individuals with $A=0$ and the same characteristics $X=x$.

To identify the more interpretable target estimand, the full data conditional risk ratio of COVID-19 in the healthcare-seeking population $RR(P_F)$, we apply Assumption \ref{ass:coretnd} such that \cite{jackson_test-negative_2013,broome_pneumococcal_1980,schnitzer_estimands_2022, jiang_double_2025}
\begin{align*}
OR(P)(x) &\overset{\textcolor{white}{A2}}{=}OR(P_F)(x)\\
&\overset{\textcolor{white}{A2}}{=}\frac{P_F(Y=1, D=1|A=1,X=x)/P_F(Y=0,D=1|A=1, X=x)}{P_F(Y=1,D=1|A=0, X=x)/P_F(Y=0,D=1|A=0, X=x)}\\
 &\overset{A\ref{ass:coretnd}}{=}\frac{P_F(Y=1, D=1|A=1,X=x)}{P_F(Y=1,D=1|A=0, X=x)}\\
&\overset{\textcolor{white}{A2}}{=}RR(P_F)(x).
\end{align*}

The first equality is from Assumptions \ref{ass:posy}-\ref{ass:conddeltaa}, \ref{ass:amisclass}, and \ref{ass:ymisclass}. The third equality is from Assumption \ref{ass:coretnd}.
\end{proof}
   
\subsection{Theorem \ref{thm:causalrr} proof}

\begin{proof}
From Theorem \ref{thm:fulldatarr}, we know that $OR(P)(x)=RR(P_F)(x)$ given Assumptions \ref{ass:posy}-\ref{ass:ymisclass}. Thus, it remains to show that $RR(P_F)(x)=RR(P_{F,ca})(x)$.
\begin{align*}
RR(P_F)(x)&\overset{\textcolor{white}{A2}}{=}\frac{P_F(Y=1, D=1|A=1,X=x)}{P_F(Y=1,D=1|A=0, X=x)}&\\
&\overset{\textcolor{white}{A2}}{=}\frac{P_{F,ca}(Y=1, D=1|A=1,X=x)}{P_{F,ca}(Y=1,D=1|A=0, X=x)}&\\
&\underset{A\ref{ass:inter}}{\overset{A\ref{ass:cons}}{=}}\frac{P_{F,ca}(Y(1)=1, D(1)=1|A=1, X=x)}{P_{F,ca}(Y(0)=1,D(0)=1|A=0, X=x)}\\
&\overset{A\ref{ass:confound}}{=}\frac{P_{F,ca}(Y(1)=1, D(1)=1|X=x)}{P_{F,ca}(Y(0)=1,D(0)=1| X=x)}\\
&\overset{\textcolor{white}{A2}}{=}RR(P_{F,ca})(x)
\end{align*}
The third line follows by Assumptions \ref{ass:cons} and \ref{ass:inter} and the fourth line follows by Assumption \ref{ass:confound}.
\end{proof}

\subsection{Semiparametric assumption}
We show that $\exp\left(\beta(P_F)^T \underline{f}(x)\right)$ in the partially linear logistic regression model can be interpreted as a full data conditional odds ratio, $ OR(P_F)(x)$.

\begin{proof}
\begin{align*}
\beta(P_F)^T \underline{f}(x)&=\text{logit} P_F(A=1|D=1,Y=1, X=x) - \text{logit} P_F(A=1|D=1,Y=0, X=x)\\
&=\log \frac{P_F(A=1|D=1,Y=1, X=x)}{P_F(A=0|D=1,Y=1, X=x) } - \log \frac{P_F(A=1|D=1,Y=0, X=x)}{ P_F(A=0|D=1,Y=0, X=x )}\\
&=\log \frac{P_F(A=1|D=1,Y=1, X=x)P_F(A=0|D=1,Y=0, X=x)}{P_F(A=0|D=1,Y=1, X=x) P_F(A=1|D=1,Y=0, X=x )}\\
&=\log\frac{P_F(Y=1,D=1|A=1, X=x)P_F(Y=0,D=1|A=0, X=x)}{P_F(Y=0,D=1|A=1, X=x) P_F(Y=1,D=1|A=0, X=x)}\\
&=\log\frac{P_F(Y=1, D=1|A=1,X=x)/P_F(Y=0,D=1|A=1, X=x)}{P_F(Y=1,D=1|A=0, X=x)/P_F(Y=0,D=1|A=0, X=x)}\\
&=\log OR(P_F)(x)
\end{align*}
The first line is by definition of the partially linear logistic regression model, the second line is by $\text{logit}$ definition, the third line is by $\log$ properties, and the fourth line is by Bayes rule. Exponentiating both sides, we obtain $\exp \left(  \beta(P_F)^T \underline{f}(x) \right)=OR(P_F)(x)$. 
\end{proof}

\subsection{Theorem \ref{thm:eif} proof}

\begin{proof}

Theorem 4 of van der Laan and Gilbert \cite{van_der_laan_semiparametric_2025} derives an efficient influence function for the $\beta(P)$ coefficient in a partially linear logistic regression model 
\begin{equation}\text{logit} \{P(J=1|R=1,A=1,W=w,\Delta = 1, T=t)\} = a\beta(P)^T \underline{f}(w,t) + h_{P}(w,t),\nonumber
\end{equation}
using the notation in van der Laan and Gilbert \cite{van_der_laan_semiparametric_2025}. Their $\beta(P)$ coefficient can be used to obtain 
\begin{equation*}
\beta(P)^T \underline{f}(x) =\log OR(P)(w,t)=\log \frac{\mu_P(1,w,t)/(1-\mu_P(1,w,t))}{\mu_P(0,w,t)/(1-\mu_P(0,w,t))},
  \end{equation*}
 where $\mu_{P}(a,w,t)= P(J=1|R=1,A=1,W=w,\Delta = 1, T=t)$. We obtain Equation \eqref{eqn:eif} from translating van der Laan and Gilbert's \cite{van_der_laan_semiparametric_2025} Theorem 4 result from their notation to our notation using Table \ref{tab:notation}. For example, van der Laan and Gilbert \cite{van_der_laan_semiparametric_2025} use $J$ and we use $A$ as the outcome of the partially lienar logistic regression model. Similarly, van der Laan and Gilbert \cite{van_der_laan_semiparametric_2025} adjust for $A$ and we adjust for a composite variable of $Y$ and $D=1$ in the partially linear logistic regression model.

\end{proof}

\subsection{Theorem \ref{thm:asympb} proof}
\begin{proof}
The asymptotic distribution for the TMLE has been proven previously in van der Laan and Gilbert \cite{van_der_laan_semiparametric_2025}. We reiterate the proof for clarity.

\begin{align*}
\beta(P_n^*)-\beta(P) &= -P_nD_{P_n^*} + (P_n-P)D_{P_n^*} + \left(\beta(P_n^*)-\beta(P)+PD_{P_n^*}\right)\\
&= o_p(n^{-1/2})+ (P_n-P)D_{P_n^*} +  o_p(n^{-1/2})\\
&= (P_n-P)D_{P_n^*} +  o_p(n^{-1/2})\\
&=  (P_n-P)D_{P}+(P_n-P)D_{P_n^*}-(P_n-P)D_{P} +  
o_p(n^{-1/2})\\
&=  (P_n-P)D_{P}+(P_n-P)(D_{P_n^*}-D_{P}) +  
o_p(n^{-1/2})\\
&=  (P_n-P)D_{P}+ o_p(n^{-1/2})+  o_p(n^{-1/2})\\
&=  (P_n-P)D_{P}+  o_p(n^{-1/2})\\
\end{align*}
The first equality is a standard expansion \cite{van_der_laan_targeted_2011}. In the second equality, we constructed the TMLE estimator such that $P_nD_{P_n^*}=o_p(n^{-1/2})$. The third term, $\beta(P_n^*)-\beta(P)+PD_{P_n^*}$, is a second-order remainder term that is $o_p(n^{-1/2})$ by Lemma 1 in van der Laan and Gilbert \cite{van_der_laan_semiparametric_2025}, which requires Conditions \ref{cond:inv}, \ref{cond:bound}, and \ref{cond:nuisrate}. In the fifth equality, $(P_n-P)(D_{P_n^*}-D_{P})=o_p(n^{-1/2})$ by Lemma 19.24 from van der Vaart \cite{van_der_vaart_asymptotic_2000}. This lemma requires that the function class containing $D_{P_n^*}-D_{P}$ is $P$-Donsker, which holds by Condition \ref{cond:donsker} and because the influence function is a Lipschitz transformation of the nuisance estimators \cite{van_der_vaart_weak_1996}.  Lemma 19.24 also requires that $|| (P_n-P)(D_{P_n^*}-D_{P})||=o_p(1)$, which is implied by the consistency of the nuisance estimators, as stated in Condition \ref{cond:nuisrate}. Given Condition \ref{cond:bound}, $D_P$ is uniformly bounded and has finite variance. From the central limit theorem, $\sqrt{n}(P_n-P)D_{P}\longrightarrow_d N\left(0, E\left[D_{P}(O)D_{P}(O)^T \right]\right)$. Lastly, we apply Slutsky's Lemma to obtain the desired result.
\end{proof}

\subsection{Theorem \ref{thm:asympor} proof}

\begin{proof}
Given Assumptions \ref{ass:posy}-\ref{ass:conddeltaa}, \ref{ass:amisclass}, and \ref{ass:ymisclass}, and Conditions \ref{cond:inv}-\ref{cond:nuisrate}, Theorem \ref{thm:asympb} states that
\[\sqrt{n} \left(\beta(P_{n}^*) - \beta(P) \right) \longrightarrow_d N\left(0, \text{cov}(D_P(O)) \right).\]

We define $g(\beta(P'))=\beta(P')^T\underline{f}(x)=\log OR(P')(x)\in\mathbb{R}$, which has gradient, \newline $\nabla g(\beta(P'))=(\frac{\partial g}{\partial \beta_0(P')},\frac{\partial g}{\partial \beta_1(P')},\cdots,\frac{\partial g}{\partial \beta_{b-1}(P')})^T = \underline{f}(x)\in\mathbb{R}^b$. We apply the multivariate delta-method to obtain the desired result:
\begin{align*}
    \sqrt{n} \left(\log OR(P_n^*)(x) -\log OR(P)(x) \right) &= \sqrt{n} \left(\beta(P_{n}^{*})^T \underline{f}(x) -\beta(P)^T \underline{f}(x) \right)\\
    &=\sqrt{n} \left(g(\beta(P_n^*)) -g(\beta(P)) \right)\\
   &\longrightarrow_d N\left(0, \nabla g (\beta(P))^T\text{ cov}(D_P(O)) \nabla g(\beta(P))\right)\\
   &=N\left(0, \underline{f}(x)^T\text{cov}(D_P(O)) \underline{f}(x)\right).
   \end{align*}
\end{proof}

% ==========  Bibliography

\bibliographystyle{vancouver}
\bibliography{tmlepaper} 

\clearpage
\hypersetup{
  colorlinks=true,
  linkcolor=black,
  filecolor=black,
  urlcolor=black,
  citecolor=black,
  pdftitle={Your Document Title},
  pdfpagemode=FullScreen,
}
\usetikzlibrary{fadings}
\usetikzlibrary{patterns}
\usetikzlibrary{shadows.blur}
\usetikzlibrary{shapes}
\renewcommand{\thetable}{S\arabic{table}}
\renewcommand{\thefigure}{S\arabic{figure}}
\renewcommand{\thesubsection}{\arabic{section}.\arabic{subsection}}
\setlist[itemize]{itemsep=1mm,parsep=1mm,topsep=0mm}
\setlist[enumerate]{itemsep=1mm,parsep=1mm,topsep=0mm}

\vspace*{.8cm} 

\begin{center}
{\LARGE Supplementary Material} \\[1em]
{\LARGE Targeted maximum likelihood estimation of vaccine effectiveness and immune correlates in test-negative design studies with missing data}\\[1em]
{ \large Leah I. B. Andrews, Lars van der Laan, and Peter B. Gilbert}
\end{center}
%\noindent * indicates corresponding author

\setcounter{table}{0}
\setcounter{figure}{0}
\setcounter{section}{0}

\section{Simulation study}
\label{sec:suppsim}

\subsection{Data-generating mechanism}

Table \ref{tab:datagentmle} and Figure \ref{fig:aydatagen} provide additional information about the data-generating mechanism of the two-phase test-negative design (TND) immune correlates simulation study described in Section \ref{subsec:datagen}.

 \begin{table}[htbp]
\caption{Data-generating mechanism for two-phase test-negative design immune correlates simulation study.}
   \label{tab:datagentmle}

    \begin{tabular}{|>{\raggedright\arraybackslash}m{0.39\linewidth}|>{\raggedright\arraybackslash}m{0.565\linewidth}|}
    \hline
         \textbf{Variable}& \textbf{Data-Generating Distribution}\\ 
         \hline
         $X_f$: \text{Female}& Resampled from COVE Distribution\\ 
        [1ex]
         $X_{co}$: \text{Comorbidities}& Resampled from COVE Distribution\\ 
      [1ex]
         $X_{t}$: \text{Calendar Date}& Resampled from COVE Distribution\\ 
         [1ex]
         $A$: High Immune Marker Level & \\
 \quad Main Effects Setting& Bernoulli$(\text{expit}(\log(0.33)+\log(3)X_{f}+\log(0.25)X_{co}+\log(1.01)X_t))$\\ 
 &\\
 \quad Interaction Setting&Bernoulli($\text{expit}(\log(0.33)+\log(3)X_{f}+\log(0.25)X_{co}+\log(1.01)X_t+\log(4)X_{f}X_{co}))$\\ 
    &\\
 \quad Splines Setting& Bernoulli$(\text{expit}(\log(0.33)+\log(3)X_{f}+\log(0.25)X_{co}+\log(1.01)X_t+\log(4)X_{f}X_{co}+\log(1.00)I_{(X_t\geq 90)}(X_t-90)+\log(0.97)I_{(X_t\geq 135)}(X_t-135)))$\\ 
 &\\
         $Y$: SARS-CoV-2 Infection& \\
 \quad Main Effects Setting& Bernoulli($\text{expit}(\log(0.15)+\beta_FA+\log(3)X_{f}+\log(4)X_{co}+\log(0.99)X_t))$\\
 &\\
 \quad Interaction Setting&Bernoulli$(\text{expit}(\log(0.15)+\beta_FA+\log(3)X_{f}+\log(4)X_{co}+\log(0.99)X_t+\log(0.25)X_{f}X_{co}))$\\ 
&\\
 \quad Splines Setting&Bernoulli$(\text{expit}(\log(0.15)+\beta_FA+\log(3)X_{f}+\log(4)X_{co}+\log(0.99)X_t+\log(0.25)X_{f}X_{co}+\log(1.00)I_{(X_t\geq 90)}(X_t-90)+\log(1.03)I_{(X_t\geq135)}(X_t-135)))$\\ 
 &\\
         $W$: Infection with a Pathogen Other than SARS-CoV-2 that Could Induce Symptoms& Bernoulli$(\text{expit}(\log(0.10)+ \log(2)X_{co}+\log(2)X_{f}+\log(1.01)X_t
+\log(0.99)I_{(X_t\geq 90)}(X_t-90)+\log(0.98)I_{(X_t\geq180)}(X_t-180)))$\\ 
&\\
         $D$: Meets Symptom Definition& Bernoulli($\text{expit}(\log(0.10)+\log(2)X_{co}
+\log(13.5)Y+\log(1.08)YX_f+\log(0.53)YX_{co}
+\log(4)W+\log(6)WX_f+\log(0.53)WX_{co}))$\\ 
&\\
         $S$: Obtains SARS-CoV-2 Test and Enrolled in TND& Sample $n$ individuals with $D=1$\\
         \hline
    \end{tabular}
\end{table}

\begin{table}[htbp]
   \addtocounter{table}{-1}
   \caption[]{Data-generating mechanism for two-phase test-negative design immune correlates simulation study (continued).}

    \begin{tabular}{|>{\raggedright\arraybackslash}m{0.37\linewidth}|>{\raggedright\arraybackslash}m{0.57\linewidth}|}
    \hline
         \textbf{Variable}& \textbf{Data-Generating Distribution}\\ 
    \hline
         $\Delta$: Observed Immune Marker Level& \\
  \quad Biased 1:1 Case-Noncase \textcolor{white}{twphase} \textcolor{white}{ii} Two-Phase Sampling Design&$\Delta=1$ for all cases and the same number of noncases sampled without replacement (40\% females without comorbidities, 10\% males without comorbidities, 10\% females with comorbidities, and 40\% males with comorbidities)\\ 
         &\\ 
         \quad Biased 1:3 Case-Noncase \textcolor{white}{twphase} \textcolor{white}{ii} Two-Phase Sampling Design& $\Delta=1$ for all cases and three times as many noncases sampled without replacement (40\% females without comorbidities, 10\% males without comorbidities, 10\% females with comorbidities, and 40\% males with comorbidities)\\ 
         &\\ 
         \quad All TND Participants& $\Delta=1$ for all cases and all noncases\\
         \hline
    \end{tabular}
    \caption*{For each Monte Carlo repetition, $N=50,000$ healthcare-seeking individuals were generated and $n = $ 500, 1000, 2000, or 3000 individuals were enrolled in the phase one TND study cohort ($S=1$). The phase two TND study cohort ($\Delta=1$) sample size depends on the two-phase sampling design and the distribution of $Y$, $X_f$, and $X_{co}$. $\beta(P_F)=\log(0.2), \log(0.7),$ and $\log(1)$.  
    \\
    Abbreviations: TND = Test-Negative Design; log = Natural Logarithm}
 
\end{table}

\begin{figure}[htb]
    \centering
    \includegraphics[width=1\linewidth]{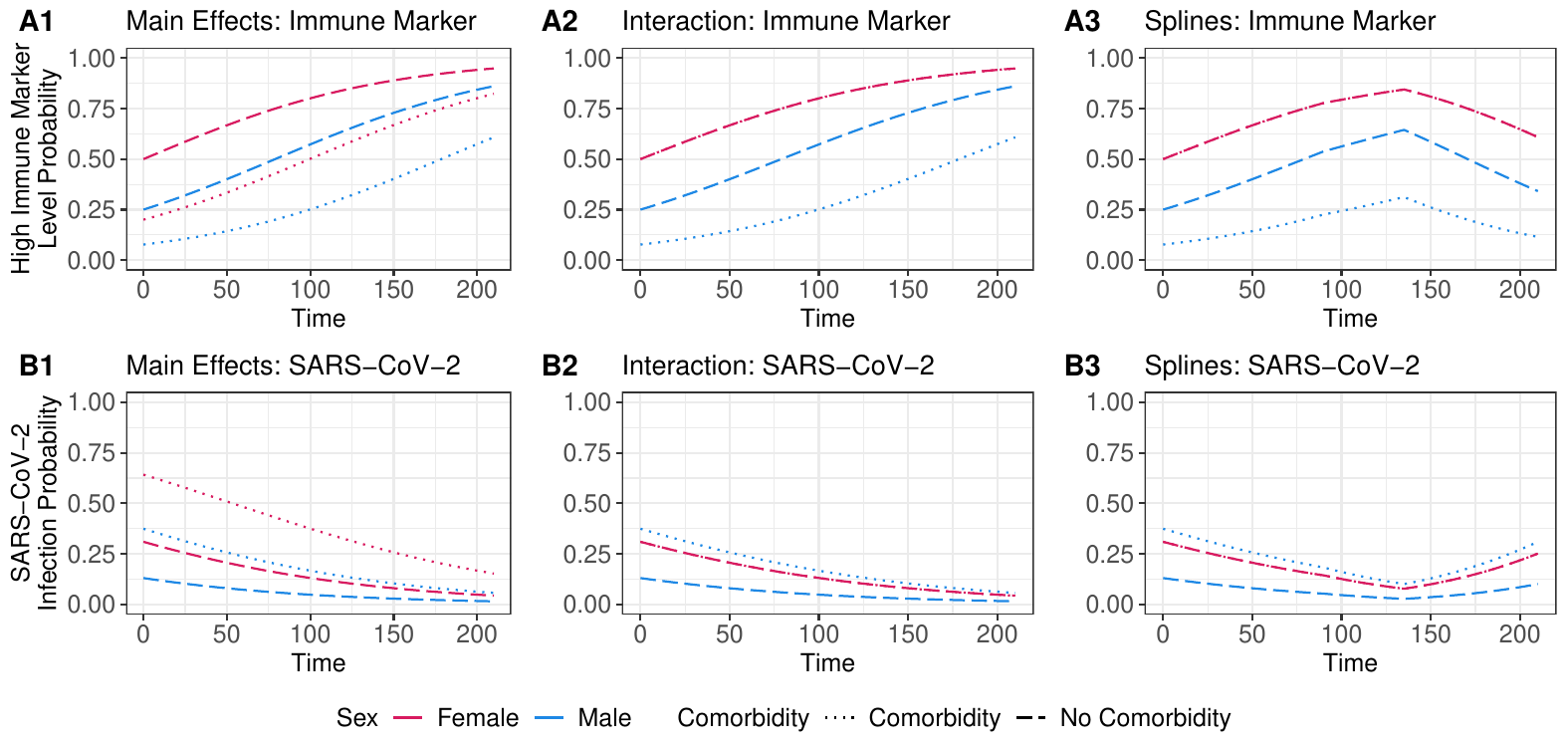}
    \caption[Data-generating distribution for high immune marker level $A$ and SARS-CoV-2 infection $Y$ under main effects, interaction, and splines confounding settings.]{Data-generating distribution for high immune marker level $A$ and SARS-CoV-2 infection $Y$ under main effects, interaction, and splines confounding settings. In the main effects setting (\textbf{A1} and \textbf{B1}), the probability of having a high immune marker level ($A=1$) and SARS-CoV-2 infection ($Y=1$) depends on sex, comorbidity, and calendar date main effects. In the interaction setting (\textbf{A2} and \textbf{B2}), $A$ and $Y$ depend on the aforementioned main effects and a sex-comorbidity interaction. In the splines setting (\textbf{A3} and \textbf{B3}), $A$ and $Y$ depend on the aforementioned main effects, sex-comorbidity interaction, and calendar date splines.}
    \label{fig:aydatagen}
\end{figure}
\newpage

\subsection{Simulation results}

We assessed bias, 95\% confidence interval (CI) coverage, type 1 error and/or power, mean estimated standard error (SE), and Monte Carlo standard deviation (MCSD) for every estimator and simulation scenario (Figures \ref{fig:biasmaineff}-\ref{fig:varcomplete} and Table S2%\ref{tab:vartmle}
).

\begin{figure}[h]
     \centering
     \includegraphics[width=1\linewidth]{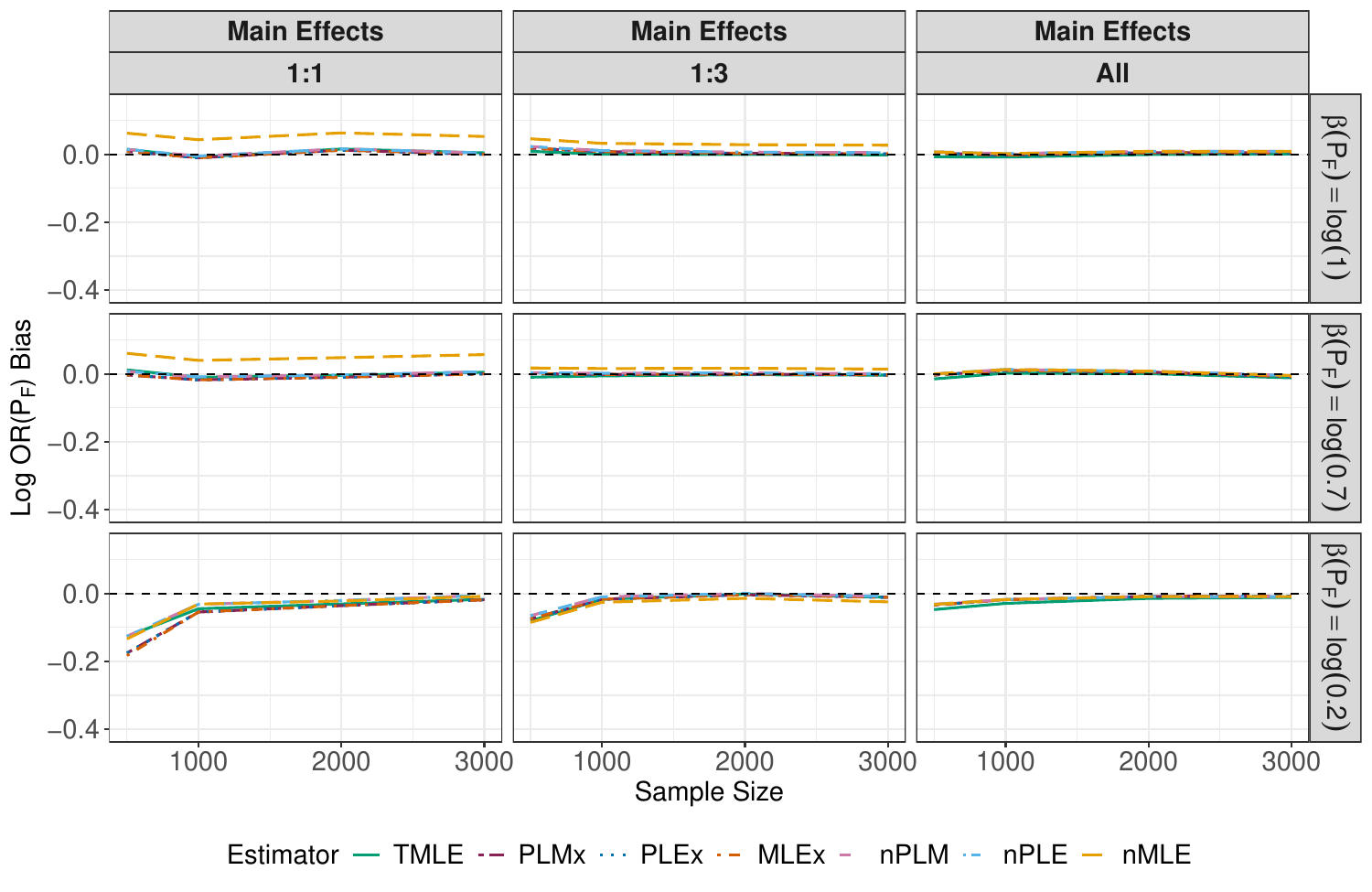}
     \caption[Bias of seven estimators from 1000 simulated two-phase test-negative design (TND) study cohorts generated under the main effects confounding setting.]{Bias of seven estimators from 1000 simulated two-phase test-negative design (TND) study cohorts generated under the main effects confounding setting. Phase one samples sizes were $500, 1000, 2000, $ and $3000$. Phase two TND participants ($A$ observed) were determined using a biased 1:1 case-noncase two-phase sampling design, biased 1:3 case-noncase two-phase sampling design, or all TND participants. The TMLE adjusts for covariates (sex, comorbidities, and calendar date) and uses highly adaptive lasso for estimation. PLMx, PLEx, and MLEx denote two pseudo-likelihood logistic regression approaches (model variance or empirical variance) and an ordinary logistic regression, respectively, that adjust for covariate main effects and a sex-comorbidity interaction. nPLM, nPLE, and nMLE denote two naïve pseudo-likelihood logistic regression approaches and a naïve ordinary logistic regression that adjust for covariate main effects.}
     \label{fig:biasmaineff}
 \end{figure}

\begin{figure}[h]
     \centering
     \includegraphics[width=1\linewidth]{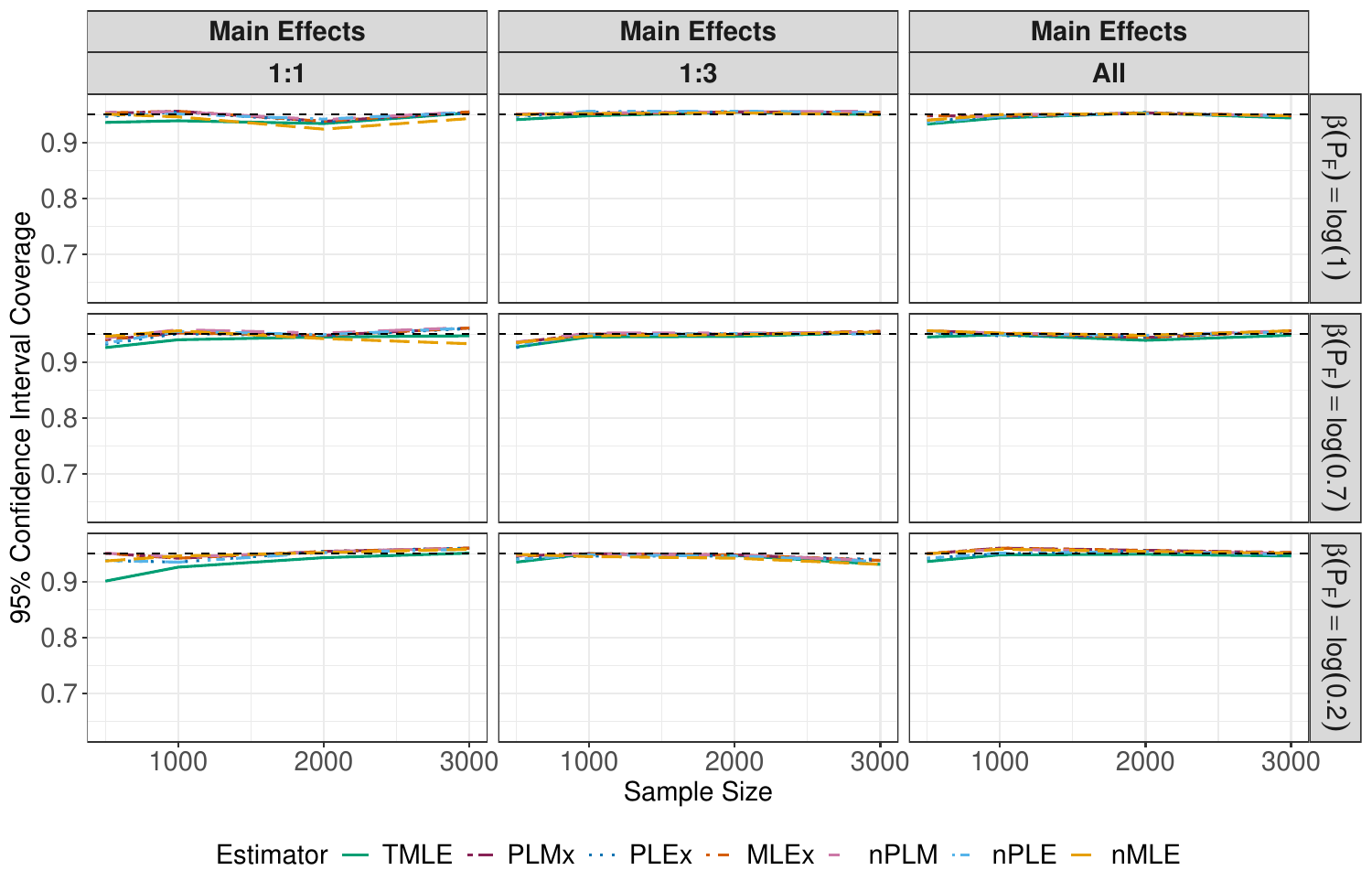}
     \caption[95\% confidence interval coverage of seven estimators from 1000 simulated two-phase test-negative design (TND) study cohorts generated under the main effects confounding setting.]{95\% confidence interval coverage of seven estimators from 1000 simulated two-phase test-negative design (TND) study cohorts generated under the main effects confounding setting. Phase one samples sizes were $500, 1000, 2000, $ and $3000$. Phase two TND participants ($A$ observed) were determined using a biased 1:1 case-noncase two-phase sampling design, biased 1:3 case-noncase two-phase sampling design, or all TND participants. The TMLE adjusts for covariates (sex, comorbidities, and calendar date) and uses highly adaptive lasso for estimation. PLMx, PLEx, and MLEx denote two pseudo-likelihood logistic regression approaches (model variance or empirical variance) and an ordinary logistic regression, respectively, that adjust for covariate main effects and a sex-comorbidity interaction. nPLM, nPLE, and nMLE denote two naïve pseudo-likelihood logistic regression approaches and a naïve ordinary logistic regression that adjust for covariate main effects.}
     \label{fig:covmaineff}
 \end{figure}

 \begin{figure}[h]
     \centering
     \includegraphics[width=1\linewidth]{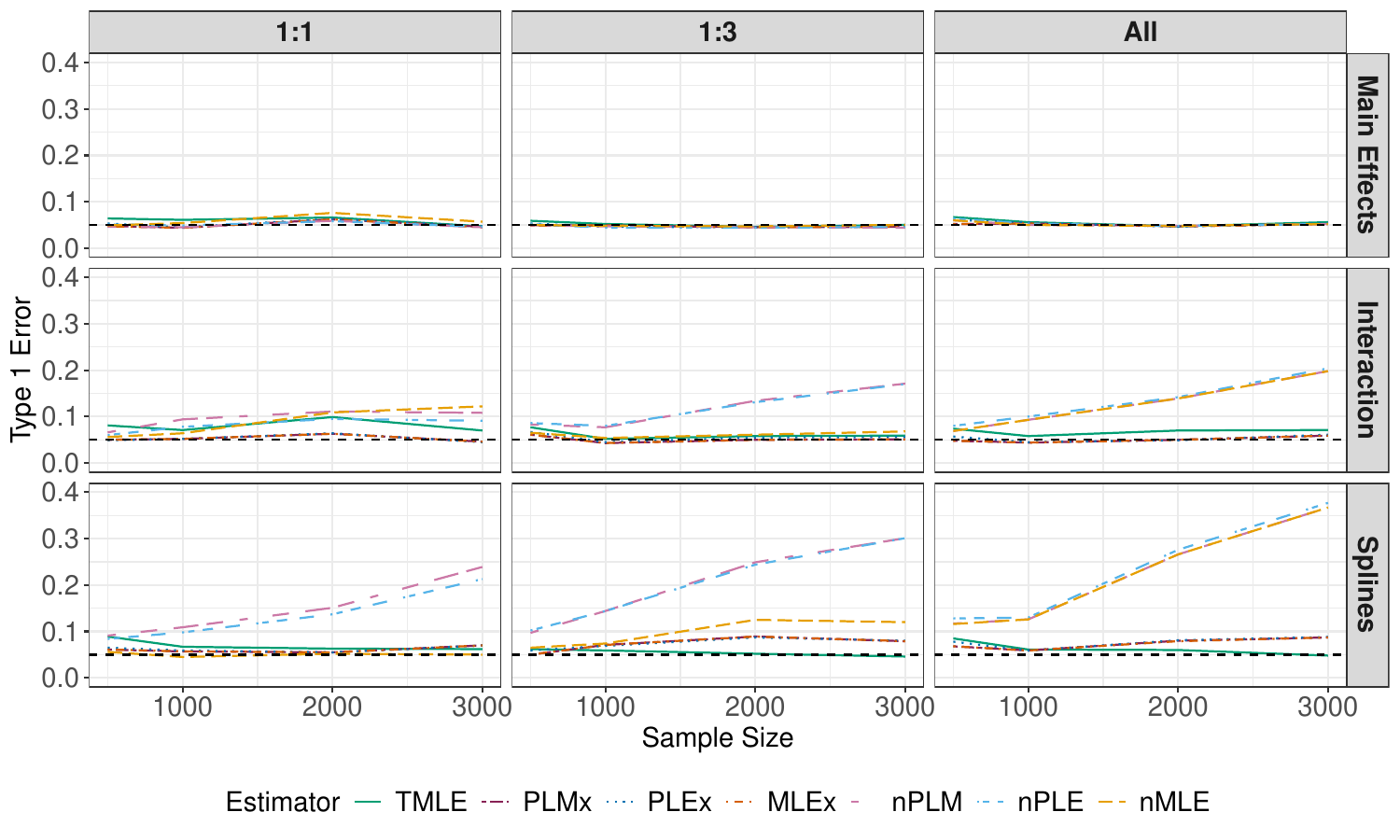}
     \caption[Type 1 error of seven estimators from 1000 simulated two-phase test-negative design (TND) study cohorts generated under the main effects, interaction, and splines confounding settings.]{Type 1 error of seven estimators from 1000 simulated two-phase test-negative design (TND) study cohorts generated under the main effects, interaction, and splines confounding settings. Phase one samples sizes were $500, 1000, 2000, $ and $3000$. Phase two TND participants ($A$ observed) were determined using a biased 1:1 case-noncase two-phase sampling design, biased 1:3 case-noncase two-phase sampling design, or all TND participants. $\beta(P_F)=\log OR (P_F)=\log(1)$. The TMLE adjusts for covariates (sex, comorbidities, and calendar date) and uses highly adaptive lasso for estimation. PLMx, PLEx, and MLEx denote two pseudo-likelihood logistic regression approaches (model variance or empirical variance) and an ordinary logistic regression, respectively, that adjust for covariate main effects and a sex-comorbidity interaction. nPLM, nPLE, and nMLE denote two naïve pseudo-likelihood logistic regression approaches and a naïve ordinary logistic regression that adjust for covariate main effects.}
     \label{fig:type1tmle}
 \end{figure}

 \begin{figure}[p]
     \centering
     \includegraphics[width=1\linewidth]{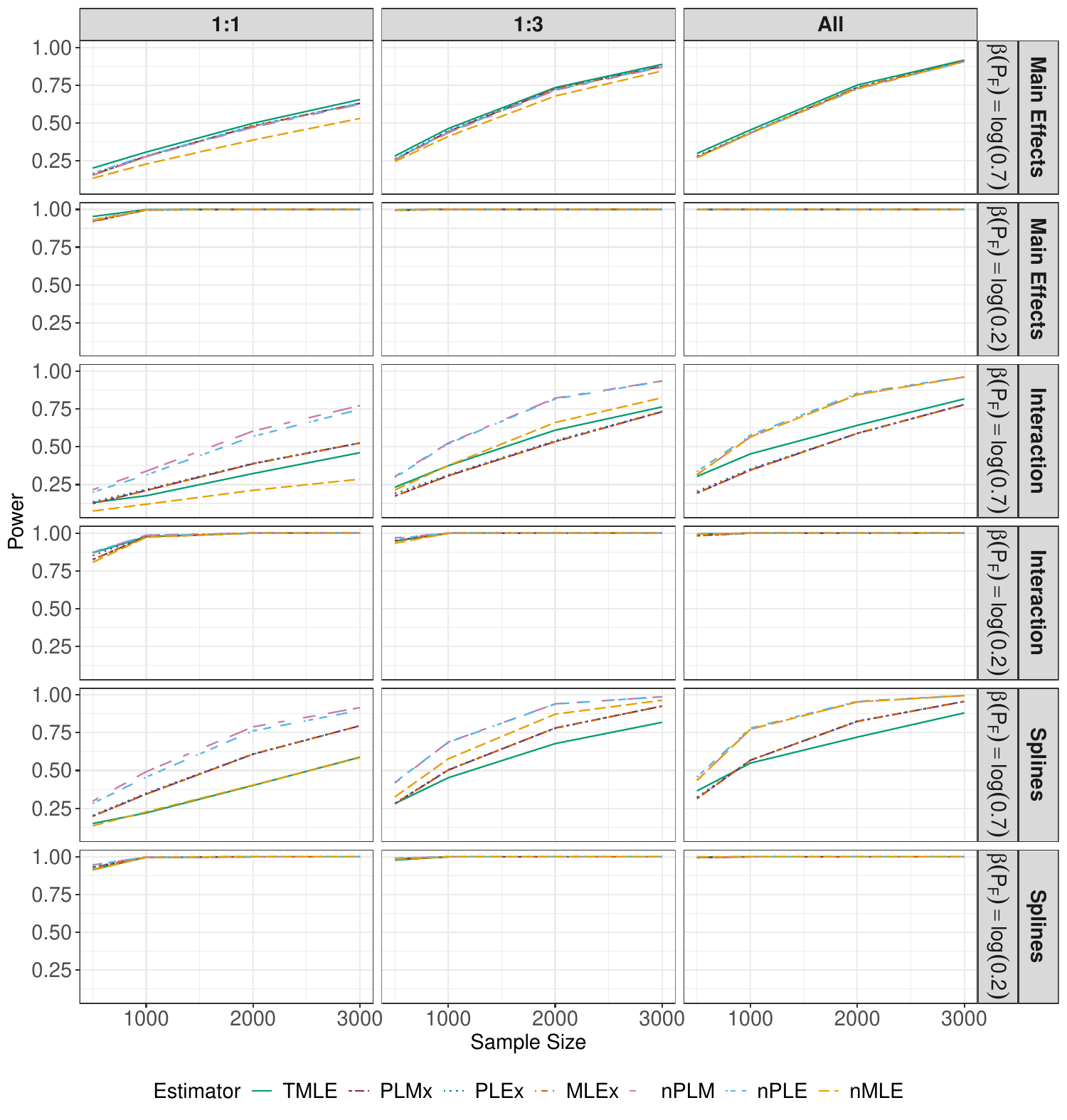}
     \caption[Power of seven estimators from 1000 simulated two-phase test-negative design (TND) study cohorts generated under the main effects, interaction, and splines confounding settings.]{Power of seven estimators from 1000 simulated two-phase test-negative design (TND) study cohorts generated under the main effects, interaction, and splines confounding settings. Phase one samples sizes were $500, 1000, 2000, $ and $3000$. Phase two TND participants ($A$ observed) were determined using a biased 1:1 case-noncase two-phase sampling design, biased 1:3 case-noncase two-phase sampling design, or all TND participants. The TMLE adjusts for covariates (sex, comorbidities, and calendar date) and uses highly adaptive lasso for estimation. PLMx, PLEx, and MLEx denote two pseudo-likelihood logistic regression approaches (model variance or empirical variance) and an ordinary logistic regression, respectively, that adjust for covariate main effects and a sex-comorbidity interaction. nPLM, nPLE, and nMLE denote two naïve pseudo-likelihood logistic regression approaches and a naïve ordinary logistic regression that adjust for covariate main effects.}
     \label{fig:powertmle}
 \end{figure}

 \begin{figure}[h]
    \centering
    \includegraphics[width=1\linewidth]{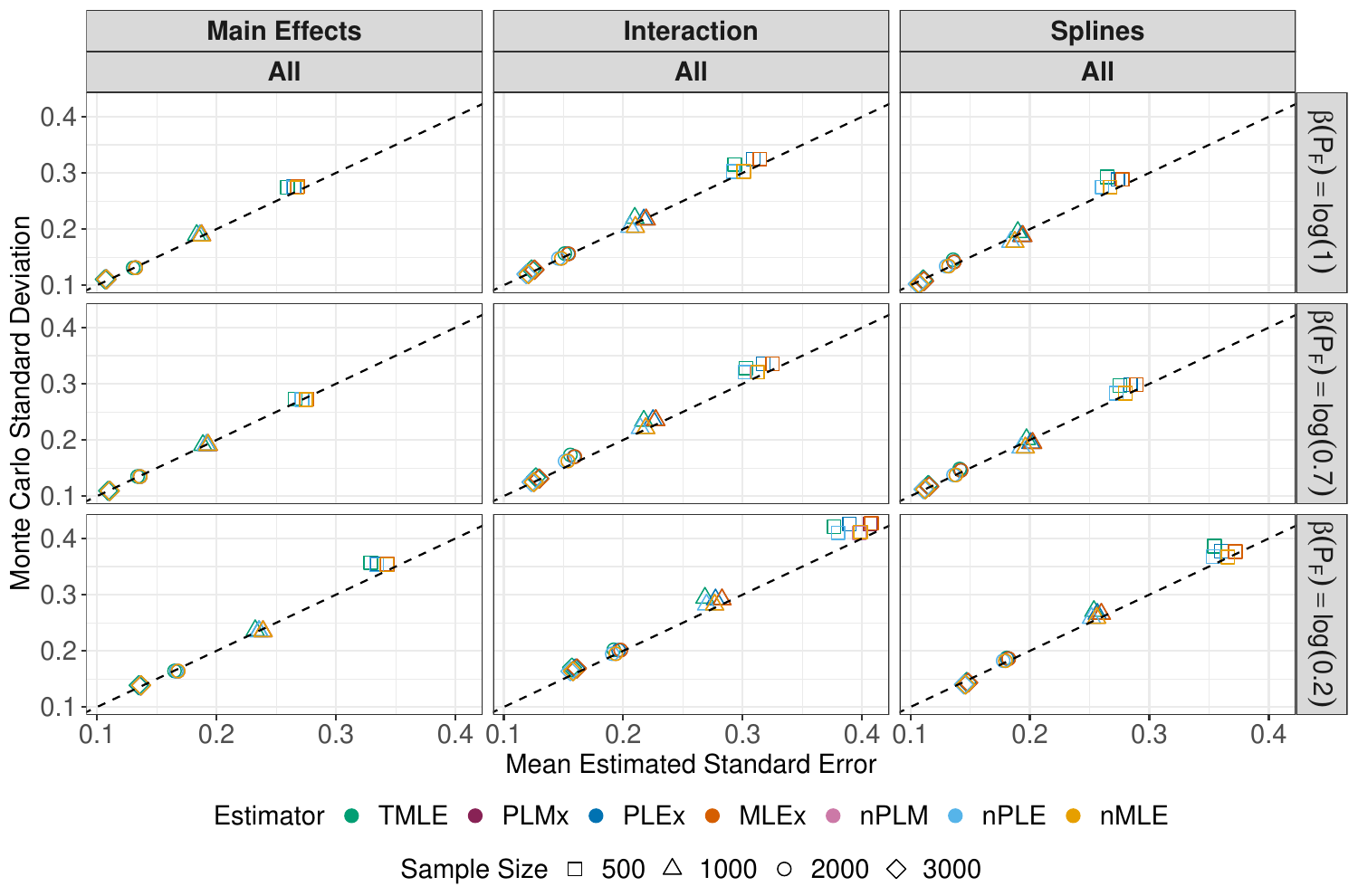}  
    \caption[Monte Carlo standard deviation and mean estimated standard error of seven estimators from 1000 simulated two-phase test-negative design (TND) study cohorts with $A$ observed on all participants, generated under three confounding settings.]{Monte Carlo standard deviation and mean estimated standard error of seven estimators from 1000 simulated two-phase test-negative design study cohorts with $A$ observed on all participants, generated under three confounding settings. The identity line is represented with black dashes. Phase one samples sizes were $500, 1000, 2000, $ and $3000$. The TMLE adjusts for covariates (sex, comorbidities, and calendar date) and uses highly adaptive lasso for estimation. PLMx, PLEx, and MLEx denote two pseudo-likelihood logistic regression approaches (model variance or empirical variance) and an ordinary logistic regression, respectively, that adjust for covariate main effects and a sex-comorbidity interaction. nPLM, nPLE, and nMLE denote two naïve pseudo-likelihood logistic regression approaches and a naïve ordinary logistic regression that adjust for covariate main effects.}
    \label{fig:varcomplete}
\end{figure}
\clearpage

\begin{longtable}{|>{\centering\arraybackslash}m{0.02\linewidth}
                  >{\centering\arraybackslash}m{0.09\linewidth} |
                  >{\centering\arraybackslash}m{0.03\linewidth}
                  >{\centering\arraybackslash}m{0.0435\linewidth}|
                  >{\centering\arraybackslash}m{0.03\linewidth}
                  >{\centering\arraybackslash}m{0.0435\linewidth}|
                  >{\centering\arraybackslash}m{0.055\linewidth}|
                  >{\centering\arraybackslash}m{0.03\linewidth}
                  >{\centering\arraybackslash}m{0.0435\linewidth}|
                  >{\centering\arraybackslash}m{0.03\linewidth}
                  >{\centering\arraybackslash}m{0.0435\linewidth}|
                  >{\centering\arraybackslash}m{0.055\linewidth}|
                  >{\centering\arraybackslash}m{0.03\linewidth}
                  >{\centering\arraybackslash}m{0.0435\linewidth}|}

\caption{Monte Carlo standard deviation and mean estimated standard error of seven logistic regression estimators from 1000 simulated two-phase test-negative design study cohorts, generated under the main effects, interaction, and splines confounding settings. \label{tab:vartmle}}\\
\hline
\multirow{2}{=}{$n$} & 2-Phase &
\multicolumn{2}{c|}{\textbf{TMLE}} &
\multicolumn{2}{c|}{\textbf{PLMx}} &
\multicolumn{1}{c|}{\textbf{PLEx}} &
\multicolumn{2}{c|}{\textbf{MLEx}} &
\multicolumn{2}{c|}{\textbf{nPLM}} &
\multicolumn{1}{c|}{\textbf{nPLE}} &
\multicolumn{2}{c|}{\textbf{nMLE}} \\
& Sampling Design &
MC SD & Mean SE &
MC SD & Mean SE &
Mean SE &
MC SD & Mean SE &
MC SD & Mean SE &
Mean SE &
MC SD & Mean SE \\
%\midrule
\hline
\endfirsthead

\caption{Monte Carlo standard deviation and mean estimated standard error of seven logistic regression estimators from 1000 simulated two-phase test-negative design study cohorts, generated under the main effects, interaction, and splines confounding settings (continued).}\\
\hline
\multirow{2}{=}{$n$} & 2-Phase &
\multicolumn{2}{c|}{\textbf{TMLE}} &
\multicolumn{2}{c|}{\textbf{PLMx}} &
\multicolumn{1}{c|}{\textbf{PLEx}} &
\multicolumn{2}{c|}{\textbf{MLEx}} &
\multicolumn{2}{c|}{\textbf{nPLM}} &
\multicolumn{1}{c|}{\textbf{nPLE}} &
\multicolumn{2}{c|}{\textbf{nMLE}} \\
& Sampling Design &
MC SD & Mean SE &
MC SD & Mean SE &
Mean SE &
MC SD & Mean SE &
MC SD & Mean SE &
Mean SE &
MC SD & Mean SE \\
\hline
\endhead

\hline
\endfoot

\endlastfoot

\multicolumn{14}{|c|}{Main Effects, $\beta(P_F) = \log(1)$}\\                                                                                      \hline       
 500   & 1:1      & 0.35            & 0.32              & 0.35            & 0.35              & 0.34              & 0.35            & 0.35              & 0.35            & 0.34              & 0.34              & 0.33            & 0.33              \\
    & 1:3      & 0.28            & 0.27              & 0.28            & 0.27              & 0.27              & 0.28            & 0.27              & 0.28            & 0.27              & 0.27              & 0.28            & 0.27              \\
    & All      & 0.27            & 0.26              & 0.28            & 0.27              & 0.26              & 0.28            & 0.27              & 0.27            & 0.27              & 0.26              & 0.27            & 0.27              \\
 1000  & 1:1      & 0.24            & 0.23              & 0.24            & 0.24              & 0.24              & 0.24            & 0.24              & 0.24            & 0.24              & 0.24              & 0.23            & 0.23              \\
   & 1:3      & 0.19            & 0.19              & 0.19            & 0.19              & 0.19              & 0.19            & 0.19              & 0.19            & 0.19              & 0.19              & 0.19            & 0.19              \\
   & All      & 0.19            & 0.18              & 0.19            & 0.19              & 0.19              & 0.19            & 0.19              & 0.19            & 0.19              & 0.19              & 0.19            & 0.19              \\
 2000  & 1:1      & 0.18            & 0.16              & 0.18            & 0.17              & 0.17              & 0.18            & 0.17              & 0.17            & 0.17              & 0.17              & 0.17            & 0.16              \\
  & 1:3      & 0.13            & 0.13              & 0.13            & 0.13              & 0.13              & 0.13            & 0.13              & 0.13            & 0.13              & 0.13              & 0.13            & 0.13              \\
  & All      & 0.13            & 0.13              & 0.13            & 0.13              & 0.13              & 0.13            & 0.13              & 0.13            & 0.13              & 0.13              & 0.13            & 0.13              \\
 3000  & 1:1      & 0.14            & 0.13              & 0.14            & 0.14              & 0.14              & 0.14            & 0.14              & 0.14            & 0.14              & 0.14              & 0.13            & 0.13              \\
   & 1:3      & 0.11            & 0.11              & 0.11            & 0.11              & 0.11              & 0.11            & 0.11              & 0.11            & 0.11              & 0.11              & 0.11            & 0.11              \\
   & All      & 0.11            & 0.11              & 0.11            & 0.11              & 0.11              & 0.11            & 0.11              & 0.11            & 0.11              & 0.11              & 0.11            & 0.11              \\
 \hline
\multicolumn{14}{|c|}{Main Effects, $\beta(P_F) = \log(0.7)$}   \\                                                                                            \hline
 500   & 1:1      & 0.41            & 0.35              & 0.41            & 0.39              & 0.38              & 0.41            & 0.39              & 0.40            & 0.38              & 0.37              & 0.38            & 0.36              \\
    & 1:3      & 0.30            & 0.28              & 0.30            & 0.29              & 0.28              & 0.30            & 0.29              & 0.30            & 0.28              & 0.28              & 0.30            & 0.28              \\
    & All      & 0.27            & 0.27              & 0.27            & 0.28              & 0.27              & 0.27            & 0.28              & 0.27            & 0.27              & 0.27              & 0.27            & 0.27              \\
 1000  & 1:1      & 0.27            & 0.25              & 0.27            & 0.27              & 0.27              & 0.27            & 0.27              & 0.26            & 0.26              & 0.26              & 0.25            & 0.25              \\
   & 1:3      & 0.20            & 0.20              & 0.20            & 0.20              & 0.20              & 0.20            & 0.20              & 0.20            & 0.20              & 0.20              & 0.20            & 0.20              \\
   & All      & 0.19            & 0.19              & 0.19            & 0.19              & 0.19              & 0.19            & 0.19              & 0.19            & 0.19              & 0.19              & 0.19            & 0.19              \\
 2000  & 1:1      & 0.19            & 0.18              & 0.19            & 0.19              & 0.19              & 0.19            & 0.19              & 0.19            & 0.19              & 0.18              & 0.18            & 0.18              \\
  & 1:3      & 0.14            & 0.14              & 0.14            & 0.14              & 0.14              & 0.14            & 0.14              & 0.14            & 0.14              & 0.14              & 0.14            & 0.14              \\
  & All      & 0.14            & 0.13              & 0.13            & 0.14              & 0.14              & 0.13            & 0.14              & 0.13            & 0.14              & 0.14              & 0.13            & 0.14              \\
 3000  & 1:1      & 0.15            & 0.15              & 0.16            & 0.15              & 0.15              & 0.16            & 0.15              & 0.15            & 0.15              & 0.15              & 0.15            & 0.15              \\
   & 1:3      & 0.11            & 0.11              & 0.11            & 0.11              & 0.11              & 0.11            & 0.11              & 0.11            & 0.11              & 0.11              & 0.11            & 0.11              \\
   & All      & 0.11            & 0.11              & 0.11            & 0.11              & 0.11              & 0.11            & 0.11              & 0.11            & 0.11              & 0.11              & 0.11            & 0.11              \\
 \hline
\multicolumn{14}{|c|}{Main Effects, $\beta(P_F) = \log(0.2)$}       \\                                                 \hline
 500   & 1:1      & 0.61            & 0.48              & 1.02            & 1.08              & 0.53              & 1.23            & 4.45              & 0.72            & 0.69              & 0.51              & 0.82            & 2.60              \\
    & 1:3      & 0.41            & 0.38              & 0.41            & 0.39              & 0.38              & 0.41            & 0.39              & 0.40            & 0.39              & 0.38              & 0.41            & 0.39              \\
    & All      & 0.36            & 0.33              & 0.35            & 0.34              & 0.33              & 0.35            & 0.34              & 0.35            & 0.34              & 0.33              & 0.35            & 0.34              \\
 1000  & 1:1      & 0.38            & 0.34              & 0.39            & 0.37              & 0.36              & 0.39            & 0.37              & 0.37            & 0.36              & 0.35              & 0.36            & 0.35              \\
   & 1:3      & 0.27            & 0.26              & 0.27            & 0.27              & 0.26              & 0.27            & 0.27              & 0.27            & 0.27              & 0.26              & 0.28            & 0.27              \\
   & All      & 0.24            & 0.23              & 0.23            & 0.24              & 0.24              & 0.23            & 0.24              & 0.23            & 0.24              & 0.24              & 0.23            & 0.24              \\
 2000  & 1:1      & 0.26            & 0.24              & 0.26            & 0.26              & 0.25              & 0.26            & 0.26              & 0.26            & 0.25              & 0.25              & 0.25            & 0.25              \\
  & 1:3      & 0.19            & 0.19              & 0.19            & 0.19              & 0.19              & 0.19            & 0.19              & 0.19            & 0.19              & 0.19              & 0.19            & 0.19              \\
  & All      & 0.16            & 0.16              & 0.16            & 0.17              & 0.17              & 0.16            & 0.17              & 0.16            & 0.17              & 0.17              & 0.16            & 0.17              \\
 3000  & 1:1      & 0.20            & 0.20              & 0.20            & 0.21              & 0.21              & 0.20            & 0.21              & 0.20            & 0.20              & 0.20              & 0.20            & 0.20              \\
   & 1:3      & 0.16            & 0.15              & 0.16            & 0.15              & 0.15              & 0.16            & 0.15              & 0.16            & 0.15              & 0.15              & 0.16            & 0.15              \\
   & All      & 0.14            & 0.13              & 0.14            & 0.14              & 0.14              & 0.14            & 0.14              & 0.14            & 0.14              & 0.14              & 0.14            & 0.14              \\
\hline
\multicolumn{14}{|c|}{Interaction, $\beta(P_F) = \log(1)$}  \\                                                               \hline                                                                                                                                                          
 500   & 1:1      & 0.42            & 0.37              & 0.42            & 0.41              & 0.41              & 0.42            & 0.41              & 0.44            & 0.41              & 0.43              & 0.40            & 0.40              \\
    & 1:3      & 0.35            & 0.32              & 0.35            & 0.33              & 0.33              & 0.35            & 0.33              & 0.34            & 0.33              & 0.32              & 0.34            & 0.33              \\
    & All      & 0.31            & 0.29              & 0.32            & 0.31              & 0.31              & 0.32            & 0.31              & 0.30            & 0.30              & 0.29              & 0.30            & 0.30              \\
 1000  & 1:1      & 0.29            & 0.26              & 0.29            & 0.28              & 0.28              & 0.29            & 0.28              & 0.31            & 0.29              & 0.30              & 0.28            & 0.27              \\
 & 1:3      & 0.23            & 0.23              & 0.23            & 0.23              & 0.23              & 0.23            & 0.23              & 0.23            & 0.23              & 0.23              & 0.22            & 0.23              \\
  & All      & 0.22            & 0.21              & 0.22            & 0.22              & 0.22              & 0.22            & 0.22              & 0.20            & 0.21              & 0.21              & 0.20            & 0.21              \\
 2000  & 1:1      & 0.21            & 0.19              & 0.21            & 0.20              & 0.20              & 0.21            & 0.20              & 0.22            & 0.20              & 0.21              & 0.20            & 0.19              \\
  & 1:3      & 0.16            & 0.16              & 0.16            & 0.16              & 0.16              & 0.16            & 0.16              & 0.16            & 0.16              & 0.16              & 0.16            & 0.16              \\
  & All      & 0.16            & 0.15              & 0.16            & 0.15              & 0.15              & 0.16            & 0.15              & 0.15            & 0.15              & 0.15              & 0.15            & 0.15              \\
 3000  & 1:1      & 0.17            & 0.16              & 0.16            & 0.16              & 0.16              & 0.16            & 0.16              & 0.17            & 0.16              & 0.17              & 0.16            & 0.16              \\
  & 1:3      & 0.14            & 0.13              & 0.13            & 0.13              & 0.13              & 0.13            & 0.13              & 0.13            & 0.13              & 0.13              & 0.13            & 0.13              \\
  & All      & 0.13            & 0.12              & 0.13            & 0.13              & 0.13              & 0.13            & 0.13              & 0.12            & 0.12              & 0.12              & 0.12            & 0.12              \\
\hline
\multicolumn{14}{|c|}{Interaction, $\beta(P_F) = \log(0.7)$} \\
\hline
 500   & 1:1      & 0.45            & 0.40              & 0.47            & 0.46              & 0.45              & 0.47            & 0.46              & 0.48            & 0.46              & 0.47              & 0.43            & 0.43              \\
    & 1:3      & 0.35            & 0.34              & 0.35            & 0.36              & 0.35              & 0.35            & 0.36              & 0.35            & 0.35              & 0.35              & 0.34            & 0.35              \\
   & All      & 0.33            & 0.30              & 0.34            & 0.32              & 0.32              & 0.34            & 0.33              & 0.32            & 0.31              & 0.30              & 0.32            & 0.31              \\
 1000  & 1:1      & 0.31            & 0.28              & 0.31            & 0.31              & 0.31              & 0.31            & 0.31              & 0.33            & 0.31              & 0.33              & 0.30            & 0.30              \\
   & 1:3      & 0.25            & 0.24              & 0.25            & 0.25              & 0.25              & 0.25            & 0.25              & 0.25            & 0.24              & 0.25              & 0.24            & 0.24              \\
   & All      & 0.23            & 0.22              & 0.24            & 0.23              & 0.22              & 0.24            & 0.23              & 0.22            & 0.22              & 0.21              & 0.22            & 0.22              \\
 2000  & 1:1      & 0.22            & 0.20              & 0.22            & 0.22              & 0.22              & 0.22            & 0.22              & 0.23            & 0.22              & 0.23              & 0.21            & 0.21              \\
   & 1:3      & 0.18            & 0.17              & 0.17            & 0.17              & 0.17              & 0.17            & 0.17              & 0.17            & 0.17              & 0.17              & 0.17            & 0.17              \\
   & All      & 0.17            & 0.16              & 0.17            & 0.16              & 0.16              & 0.17            & 0.16              & 0.16            & 0.15              & 0.15              & 0.16            & 0.15              \\
 3000  & 1:1      & 0.19            & 0.17              & 0.18            & 0.18              & 0.18              & 0.18            & 0.18              & 0.19            & 0.18              & 0.19              & 0.17            & 0.17              \\
   & 1:3      & 0.14            & 0.14              & 0.13            & 0.14              & 0.14              & 0.13            & 0.14              & 0.13            & 0.14              & 0.14              & 0.13            & 0.14              \\
  & All      & 0.13            & 0.13              & 0.13            & 0.13              & 0.13              & 0.13            & 0.13              & 0.13            & 0.13              & 0.12              & 0.13            & 0.13              \\
\hline
\multicolumn{14}{|c|}{Interaction, $\beta(P_F) = \log(0.2)$}  \\
\hline
 500   & 1:1      & 0.75            & 0.58              & 1.51            & 2.23              & 0.66              & 1.95            & 17.45             & 0.89            & 0.92              & 0.67              & 0.90            & 2.62              \\
   & 1:3      & 0.52            & 0.47              & 0.53            & 0.50              & 0.49              & 0.53            & 0.50              & 0.52            & 0.49              & 0.49              & 0.52            & 0.49              \\
  & All      & 0.42            & 0.38              & 0.43            & 0.41              & 0.39              & 0.43            & 0.41              & 0.41            & 0.40              & 0.38              & 0.41            & 0.40              \\
 1000  & 1:1      & 0.45            & 0.40              & 0.45            & 0.44              & 0.44              & 0.45            & 0.44              & 0.46            & 0.44              & 0.45              & 0.42            & 0.43              \\
   & 1:3      & 0.34            & 0.33              & 0.34            & 0.34              & 0.34              & 0.34            & 0.34              & 0.34            & 0.33              & 0.34              & 0.34            & 0.34              \\
   & All      & 0.29            & 0.27              & 0.29            & 0.28              & 0.28              & 0.29            & 0.28              & 0.28            & 0.28              & 0.27              & 0.28            & 0.28              \\
 2000  & 1:1      & 0.31            & 0.28              & 0.32            & 0.31              & 0.30              & 0.32            & 0.31              & 0.32            & 0.30              & 0.32              & 0.30            & 0.29              \\
  & 1:3      & 0.24            & 0.23              & 0.24            & 0.24              & 0.23              & 0.24            & 0.24              & 0.24            & 0.23              & 0.24              & 0.24            & 0.23              \\
   & All      & 0.20            & 0.19              & 0.20            & 0.20              & 0.20              & 0.20            & 0.20              & 0.19            & 0.19              & 0.19              & 0.19            & 0.19              \\
 3000  & 1:1      & 0.25            & 0.23              & 0.25            & 0.25              & 0.24              & 0.25            & 0.25              & 0.26            & 0.24              & 0.25              & 0.24            & 0.24              \\
  & 1:3      & 0.20            & 0.19              & 0.20            & 0.19              & 0.19              & 0.20            & 0.19              & 0.20            & 0.19              & 0.19              & 0.20            & 0.19              \\
   & All      & 0.17            & 0.16              & 0.17            & 0.16              & 0.16              & 0.17            & 0.16              & 0.16            & 0.16              & 0.16              & 0.16            & 0.16              \\
\hline
\multicolumn{14}{|c|}{Splines, $\beta(P_F) = \log(1)$}  \\                                                                \hline                                                                                                                                                           
 500   & 1:1      & 0.38            & 0.34              & 0.38            & 0.37              & 0.37              & 0.38            & 0.37              & 0.39            & 0.37              & 0.38              & 0.36            & 0.36              \\
   & 1:3      & 0.30            & 0.28              & 0.30            & 0.29              & 0.29              & 0.30            & 0.29              & 0.29            & 0.28              & 0.28              & 0.29            & 0.29              \\
   & All      & 0.29            & 0.26              & 0.29            & 0.28              & 0.27              & 0.29            & 0.28              & 0.27            & 0.27              & 0.26              & 0.27            & 0.27              \\
 1000  & 1:1      & 0.26            & 0.24              & 0.25            & 0.25              & 0.25              & 0.25            & 0.25              & 0.27            & 0.26              & 0.27              & 0.24            & 0.25              \\
   & 1:3      & 0.21            & 0.20              & 0.21            & 0.20              & 0.20              & 0.21            & 0.20              & 0.20            & 0.20              & 0.20              & 0.20            & 0.20              \\
  & All      & 0.19            & 0.19              & 0.19            & 0.19              & 0.19              & 0.19            & 0.19              & 0.18            & 0.19              & 0.18              & 0.18            & 0.19              \\
 2000  & 1:1      & 0.18            & 0.18              & 0.18            & 0.18              & 0.18              & 0.18            & 0.18              & 0.19            & 0.18              & 0.19              & 0.17            & 0.17              \\
  & 1:3      & 0.14            & 0.14              & 0.14            & 0.14              & 0.14              & 0.14            & 0.14              & 0.14            & 0.14              & 0.14              & 0.14            & 0.14              \\
   & All      & 0.15            & 0.14              & 0.14            & 0.14              & 0.14              & 0.14            & 0.14              & 0.13            & 0.13              & 0.13              & 0.13            & 0.13              \\
 3000  & 1:1      & 0.15            & 0.14              & 0.14            & 0.15              & 0.15              & 0.14            & 0.15              & 0.15            & 0.15              & 0.15              & 0.14            & 0.14              \\
  & 1:3      & 0.12            & 0.12              & 0.11            & 0.12              & 0.12              & 0.11            & 0.12              & 0.11            & 0.11              & 0.12              & 0.11            & 0.12              \\
   & All      & 0.11            & 0.11              & 0.11            & 0.11              & 0.11              & 0.11            & 0.11              & 0.10            & 0.11              & 0.11              & 0.10            & 0.11              \\
\hline
\multicolumn{14}{|c|}{Splines, $\beta(P_F) = \log(0.7)$} \\                                                            \hline
 500   & 1:1      & 0.40            & 0.37              & 0.41            & 0.40              & 0.40              & 0.41            & 0.40              & 0.42            & 0.40              & 0.41              & 0.38            & 0.39              \\
   & 1:3      & 0.33            & 0.30              & 0.32            & 0.31              & 0.31              & 0.32            & 0.31              & 0.32            & 0.31              & 0.31              & 0.32            & 0.31              \\
   & All      & 0.30            & 0.28              & 0.30            & 0.29              & 0.28              & 0.30            & 0.29              & 0.28            & 0.28              & 0.27              & 0.28            & 0.28              \\
 1000  & 1:1      & 0.28            & 0.26              & 0.28            & 0.28              & 0.28              & 0.28            & 0.28              & 0.29            & 0.28              & 0.29              & 0.26            & 0.27              \\
   & 1:3      & 0.23            & 0.22              & 0.22            & 0.22              & 0.22              & 0.22            & 0.22              & 0.22            & 0.21              & 0.21              & 0.22            & 0.22              \\
  & All      & 0.20            & 0.20              & 0.19            & 0.20              & 0.20              & 0.19            & 0.20              & 0.19            & 0.20              & 0.19              & 0.19            & 0.20              \\
 2000  & 1:1      & 0.20            & 0.19              & 0.19            & 0.19              & 0.19              & 0.19            & 0.20              & 0.20            & 0.20              & 0.20              & 0.18            & 0.19              \\
   & 1:3      & 0.16            & 0.15              & 0.16            & 0.15              & 0.15              & 0.16            & 0.15              & 0.15            & 0.15              & 0.15              & 0.15            & 0.15              \\
   & All      & 0.15            & 0.14              & 0.15            & 0.14              & 0.14              & 0.15            & 0.14              & 0.14            & 0.14              & 0.14              & 0.14            & 0.14              \\
 3000  & 1:1      & 0.16            & 0.16              & 0.16            & 0.16              & 0.16              & 0.16            & 0.16              & 0.17            & 0.16              & 0.17              & 0.16            & 0.15              \\
   & 1:3      & 0.13            & 0.13              & 0.12            & 0.13              & 0.13              & 0.12            & 0.13              & 0.12            & 0.12              & 0.12              & 0.12            & 0.12              \\
   & All      & 0.12            & 0.12              & 0.12            & 0.12              & 0.12              & 0.12            & 0.12              & 0.11            & 0.11              & 0.11              & 0.11            & 0.11              \\
\hline
\multicolumn{14}{|c|}{Splines, $\beta(P_F) = \log(0.2)$}    \\                                                            \hline                                                                                                                                                
 500   & 1:1      & 0.65            & 0.52              & 0.87            & 0.96              & 0.56              & 1.17            & 6.33              & 0.75            & 0.77              & 0.56              & 0.97            & 3.58              \\
   & 1:3      & 0.46            & 0.43              & 0.45            & 0.43              & 0.42              & 0.45            & 0.43              & 0.44            & 0.42              & 0.42              & 0.45            & 0.43              \\
   & All      & 0.39            & 0.35              & 0.38            & 0.37              & 0.36              & 0.38            & 0.37              & 0.37            & 0.37              & 0.35              & 0.37            & 0.37              \\
 1000  & 1:1      & 0.40            & 0.36              & 0.40            & 0.38              & 0.38              & 0.40            & 0.38              & 0.41            & 0.38              & 0.39              & 0.40            & 0.37              \\
   & 1:3      & 0.33            & 0.30              & 0.32            & 0.30              & 0.29              & 0.32            & 0.30              & 0.31            & 0.29              & 0.29              & 0.32            & 0.30              \\
   & All      & 0.27            & 0.25              & 0.27            & 0.26              & 0.26              & 0.27            & 0.26              & 0.26            & 0.26              & 0.25              & 0.26            & 0.26              \\
 2000  & 1:1      & 0.27            & 0.26              & 0.27            & 0.27              & 0.26              & 0.27            & 0.27              & 0.27            & 0.26              & 0.27              & 0.26            & 0.26              \\
   & 1:3      & 0.21            & 0.21              & 0.20            & 0.21              & 0.21              & 0.20            & 0.21              & 0.20            & 0.20              & 0.21              & 0.20            & 0.21              \\
   & All      & 0.19            & 0.18              & 0.19            & 0.18              & 0.18              & 0.19            & 0.18              & 0.18            & 0.18              & 0.18              & 0.18            & 0.18              \\
 3000  & 1:1      & 0.22            & 0.21              & 0.22            & 0.22              & 0.21              & 0.22            & 0.22              & 0.22            & 0.21              & 0.22              & 0.21            & 0.21              \\
  & 1:3      & 0.17            & 0.17              & 0.17            & 0.17              & 0.17              & 0.17            & 0.17              & 0.17            & 0.17              & 0.17              & 0.17            & 0.17              \\
  & All      & 0.14            & 0.15              & 0.14            & 0.15              & 0.15              & 0.14            & 0.15              & 0.14            & 0.15              & 0.14              & 0.14            & 0.15   \\

\hline
\end{longtable}
\vspace{-5mm}
\noindent
Phase two TND participants ($A$ observed) were determined using a biased 1:1 case-noncase two-phase sampling design, biased 1:3 case-noncase two-phase sampling design, or all TND participants. The TMLE adjusts for covariates (sex, comorbidities, and calendar date) and uses highly adaptive lasso for estimation. PLMx, PLEx, and MLEx denote two pseudo-likelihood logistic regression approaches (model variance or empirical variance) and an ordinary logistic regression, respectively, that adjust for covariate main effects and an interaction. nPLM, nPLE, and nMLE denote two naïve pseudo-likelihood logistic regression approaches and a naïve ordinary logistic regression that adjust for covariate main effects. The Monte Carlo standard deviations for PLEx and nPLE were identical to the Monte Carlo standard deviations for PLMx and nPLE, respectively, and omitted from the table.
\\
Abbreviations: TMLE = Targeted Maximum Likelihood Estimator; PL = Pseudo-Likelihood; MLE = Maximum Likelihood Estimator; MCSD = Monte Carlo Standard Deviation; SE = Standard Error; TND = Test-Negative Design

\section{Application 1: inference on COVID-19 vaccine effectiveness}
\label{sec:suppapp1}

The primary COVID-19 symptom definition studied in COVE was defined as having at least two systemic symptoms (i.e., fever, chills, myalgia, headache, sore throat, new olfactory and taste disorder) or at least one respiratory symptom (i.e., cough, shortness or breath or difficulty breathing, pneumonia) \cite{el_sahly_efficacy_2021}.

In COVE, participants' age was measured in years and sex was measured as male or female at birth. COVE collected detailed self-reported race and ethnicity information, which we used to define a binary race and ethnicity indicator, Person of Color vs. Non-Hispanic/Latino White, to simplify covariate adjustment and manage small subgroup sizes when adjusting for additional covariates. We defined a Person of Color as a participant who reported their race as American Indian or Alaska Native, Asian, Black or African American, Multiple, Native Hawaiian or Other Pacific Islander, or other, and/or who reported their ethnicity as Hispanic or Latino. A Non-Hispanic/Latino White individual was defined as a participant who reported their race as White and either reported their ethnicity as Not Hispanic or Latino or did not report their ethnicity. Participants who did not report their race and either did not report their ethnicity or reported their ethnicity as Not Hispanic or Latino were classified as missing and excluded from analyses. We defined comorbidities as having at least one risk factor for severe COVID-19 (i.e., chronic lung disease, cardiac disease, severe obesity, diabetes, liver disease, or infection with human immunodeficiency virus) \cite{el_sahly_efficacy_2021}. We defined region as the four United States Census Groups: Midwest, Northeast, South, and West \cite{us_census_bureau_geographic_2021}. We also converted participants' SARS-CoV-2 testing date into a quantitative calendar date variable using September 1, 2020 as a reference date (for TMLE) or a categorical two-week variable (for ordinary logistic regression) \cite{lopez_bernal_effectiveness_2021, chua_use_2020,belongia_effectiveness_2009, bond_regression_2016}. For the TMLE approach, we standardized age and calendar date to have mean zero and standard deviation of one.

\begin{table}[!htbp]
\caption[Characteristics of Moderna COVE participants included in the TND vaccine effectiveness study cohort.]{Characteristics of Moderna COVE participants included in the TND vaccine effectiveness study cohort.}
\label{tab:vetmletbl1}
\begin{tabular}{|>{\raggedright\arraybackslash}m{0.5\linewidth}|>{\centering\arraybackslash}p{0.2\linewidth}|>{\centering\arraybackslash}p{0.2\linewidth}|}

\hline

 \parbox[c][1cm][c]{0.5\linewidth}{\centering \textbf{Characteristic}} &  \textbf{Case} \par  $n = 728$&  \textbf{Noncase}  \par $n = 1,825$\\

\hline
Age (Years) & 48 (14) & 48 (15)\\
Female & 354 (49\%) & 1,061 (58\%)\\
Person of Color & 247 (34\%) & 630 (35\%)\\
At Least One Comorbidity & 178 (24\%) & 458 (25\%)\\
United States Census Region &  & \\
\quad Midwest & 185 (25\%) & 388 (21\%)\\
\quad Northeast & 40 (5.5\%) & 116 (6.4\%)\\
\quad South & 360 (49\%) & 910 (50\%)\\
\quad West & 143 (20\%) & 411 (23\%)\\

Testing Date (Days Since 09-01-20) & 97 (28) & 84 (32)\\

Testing Date (Two-Week Interval) &  & \\

\quad 09-07-20 & 0 (0\%) & 16 (0.9\%)\\

\quad 09-21-20 & 15 (2.1\%) & 80 (4.4\%)\\

\quad 10-05-20 & 18 (2.5\%) & 160 (8.8\%)\\

\quad 10-19-20 & 44 (6.0\%) & 201 (11\%)\\

\quad 11-02-20 & 95 (13\%) & 271 (15\%)\\

\quad 11-16-20 & 91 (13\%) & 275 (15\%)\\

\quad 11-30-20 & 138 (19\%) & 309 (17\%)\\

\quad 12-14-20 & 105 (14\%) & 185 (10\%)\\

\quad 12-28-20 & 155 (21\%) & 180 (9.9\%)\\

\quad 01-11-21 & 53 (7.3\%) & 89 (4.9\%)\\

\quad 01-25-21 & 12 (1.6\%) & 45 (2.5\%)\\

\quad 02-08-21 & 1 (0.1\%) & 10 (0.5\%)\\

\quad 02-22-21 & 1 (0.1\%) & 4 (0.2\%)\\
Received mRNA-1273 Vaccine & 46 (6.3\%) & 866 (47\%)\\
%\bottomrule
\hline
\end{tabular}
\caption*{Mean (standard deviation) is reported for quantitative variables and number of participants (percentage) is reported for categorical variables. All covariates except testing date were recorded at randomized placebo-controlled clinical trial enrollment. A Person of Color was defined as a participant who reported their race as American Indian or Alaska Native, Asian, Black or African American, Multiple, Native Hawaiian or Other Pacific Islander, or Other, and/or reported their ethnicity as Hispanic or Latino. A Non-Hispanic/Latino White individual was defined as a participant who reported their race as White and either reported their ethnicity as Not Hispanic or Latino or did not report their ethnicity. Participants who did not report their race and either did not report their ethnicity or reported their ethnicity as Not Hispanic or Latino were excluded from the TND study cohort.}

\end{table}

\begin{table}[!htbp]
\caption{Candidate ensemble super learner algorithms for targeted maximum likelihood estimator to assess COVID-19 vaccine effectiveness.}
\label{tab:vesl}

    \begin{tabular}{|>{\raggedright\arraybackslash}m{0.31\linewidth}|>{\centering\arraybackslash}p{0.16\linewidth}|>{\centering\arraybackslash}m{0.42\linewidth}|}

    \hline
          \textbf{Learning Algorithm}&\textbf{sl3 Learner}& \textbf{\textcolor{white}{sp}Algorithm Tuning Parameters \textcolor{white}{sp} and Screens}\\
          \hline
          Sample Mean&Lrnr\_mean& screen = none\\
          Logistic Regression with Main Effects&Lrnr\_glm&  screen = none, lasso\\
          Logistic Regression with Main Effects and Interaction&Lrnr\_glm&  screen = none, lasso\\
          General Additive Models&Lrnr\_gam&  screen = none, lasso\\
           Multivariate Adaptive Regression Splines&Lrnr\_earth& degree = 2; screen = none\\
          Lasso Regression&Lrnr\_glmnet& alpha = 1;  screen = none\\
          Ridge Regression&Lrnr\_glmnet& alpha = 0;  screen = none\\
          Random Forest&Lrnr\_ranger& screen = none\\
          Gradient Boosting&Lrnr\_xgboost& max\_depth = 3,5; eta = 0.3; \par scale\_pos\_weight = 1; screen = none\\
 Highly Adaptive Lasso& Lrnr\_hal9001&smoothness\_orders = 1; max\_degree = 1,2; num\_knots = 3, 10; screen = none\\
 %\bottomrule
 \hline
    \end{tabular}
\caption*{Learners from the sl3 R package \cite{coyle_sl3_2021} considered for estimating the probability of vaccination in individuals without SARS-CoV-2 infection, given age, sex, race and ethnicity, comorbidities, region, and calendar date and the probability of SARS-CoV-2 infection given the aforementioned covariates. Age and calendar date were standardized with mean zero and standard deviation of one. Learners either considered all covariates (no screen) or only covariates with nonzero coefficients based on lasso regression using 10-fold cross-validation (lasso screen). Candidate learners were assessed using a binomial log-likelihood loss and $10$-fold cross validation and predictions were combined using non-negative least squares regression. 10-fold cross-fitting was used to obtain the targeted maximum likelihood estimator.}
\end{table}

In the TMLE, we used a 10-fold cross-fitted ensemble super learner \cite{van_der_laan_super_2007, chernozhukov_doubledebiased_2018} that considered the sample mean, logistic regression models, general additive models, multivariate adaptive regression splines, lasso regression, ridge regression, highly adaptive lasso, random forests, and gradient boosting for estimation (Table \ref{tab:vesl}). The ensemble super learner also considered screened logistic regression and general additive model learners, which only included covariates with nonzero coefficients based on lasso regression. We assessed the candidate learners using a binomial log-likelihood loss and 10-fold cross-validation and combined the nuisance function predictions using non-negative least squares regression.

\section{Application 2: inference on immune markers as exposure-proximal correlates of COVID-19}
\label{sec:suppapp2}

The COVID-19 symptom definition recommended by the CDC in 2020 required at least one of the following symptoms: fever, chills, cough, shortness of breath or difficulty breathing, fatigue, muscle aches or body aches, headache, new loss of taste or smell, sore throat, nasal congestion or rhinorrhea, nausea or vomiting, or diarrhea \cite{cdc_symptoms_2022, stokes_coronavirus_2020}. 

Previous COVE immune correlates studies collected several immune markers from a subset of COVE participants \cite{gilbert_immune_2022, follmann_test-negative_2025}. Participants' 50\% inhibitory dilution neutralizing antibody (nAb) titers were measured using the Monogram Biosciences PhenoSense$^{TM}$ pseudovirus D614G neutralization assay, measured in World Health Organization (WHO) international units (IU50/ml) \cite{follmann_test-negative_2025}.  The limit of detection for nAbs assays was 2.61 IU50/ml. Participants' immunoglobulin G (IgG) binding antibody (bAb) concentration measured against the ancestral SARS-CoV-2 spike protein, receptor-binding domain (RBD), and nucleocapsid antigens were quantified using a Meso Scale Discovery 4-plex assay with measurements in WHO international binding antibody units per milliliter (BAU/ml) \cite{follmann_test-negative_2025,gilbert_immune_2022}. The limits of detection for anti-spike, anti-RBD, and anti-nucleocapsid IgG bAb assays were 0.31 BAU/ml, 1.59 BAU/ml, and 0.09 BAU/ml, respectively, and the positivity threshold for anti-nucleocapsid IgG bAb assays was 23.5 BAU/ml  \cite{gilbert_immune_2022, follmann_test-negative_2025}. 

We derived our phase two TND study cohort from data collected from previous COVE immune correlates studies \cite{follmann_test-negative_2025,gilbert_immune_2022}. Follmann et al. \cite{follmann_test-negative_2025} selected 186 vaccinated and symptomatic COVE participants who obtained SARS-CoV-2 testing during the blinded phase to measure immune markers on. Cases were matched to roughly three noncases by United States Census region and testing date within seven days. Based on our eligibility criteria in Section \ref{subsec:appimm}, our phase one TND cohort consisted of 935 COVE participants (46 cases and 889 noncases), with 165 participants (44 cases, 121 noncases) selected to have their immune markers measured in Follmann et al. \cite{follmann_test-negative_2025}. Of these 165 participants, we excluded four participants missing anti-spike IgG bAb concentrations, anti-RBD IgG bAb concentrations, and/or nAb titers; eight participants who were anti-nucleocapsid seropositive; and all noncases matched to cases who did not meet the above criteria. Our resulting phase two TND study cohort consisted of 127 participants (34 cases and 93 noncases) from Follmann et al. \cite{follmann_test-negative_2025}.

\begin{table}
\caption[Characteristics of Moderna COVE participants included in the phase one and phase two test-negative design immune correlates study cohorts.]{Characteristics of Moderna COVE participants included in the phase one and phase two test-negative design immune correlates study cohorts.}
\label{tab:immtmletbl1}
\begin{tabular}{|>{\raggedright\arraybackslash}m{0.39\linewidth}|>{\centering\arraybackslash}m{0.12\linewidth}>{\centering\arraybackslash}m{0.13\linewidth}|>{\centering\arraybackslash}m{0.12\linewidth}>{\centering\arraybackslash}m{0.13\linewidth}|}
\hline
\multirow{3}{=}{\textbf{Characteristic}} & 
\multicolumn{2}{c|}{\textbf{Phase One}} &\multicolumn{2}{c|}{\textbf{Phase Two}} \\
& \textbf{Case}  & \textbf{Noncase} &  \textbf{Case} &\textbf{Noncase}  \\
 & \quad $n = 46$ &$n = 889$ & $n = 34$&$n = 93$\\
\hline
Age (Years) & 48 (14) & 49 (15) & 48 (14) & 51 (15)\\

Female & 22 (48\%) & 523 (59\%) & 16 (47\%) & 53 (57\%)\\

Person of Color & 11 (24\%) & 318 (36\%) & 9 (26\%) & 44 (47\%)\\

At Least One Comorbidity & 13 (28\%) & 234 (26\%) & 9 (26\%) & 26 (28\%)\\

United States Census Region &  &  &  & \\

\quad Midwest & 8 (17\%) & 187 (21\%) & 6 (18\%) & 17 (18\%)\\

\quad Northeast & 3 (6.5\%) & 57 (6.4\%) & 2 (5.9\%) & 6 (6.5\%)\\

\quad South & 24 (52\%) & 417 (47\%) & 17 (50\%) & 45 (48\%)\\

\quad West & 11 (24\%) & 228 (26\%) & 9 (26\%) & 25 (27\%)\\

Testing Date (Days Since 09-01-20) & 106 (27) & 95 (29) & 112 (24) & 116 (22)\\

Testing Date (Phase One Tertiles) &  &  &  & \\

\quad 10-02-20 to 11-20-20 & 10 (22\%) & 307 (35\%) & 5 (15\%) & 9 (9.7\%)\\

\quad 11-21-20 to 12-17-20 & 9 (20\%) & 298 (34\%) & 7 (21\%) & 20 (22\%)\\

\quad 12-18-20 to 02-26-21 & 27 (59\%) & 284 (32\%) & 22 (65\%) & 64 (69\%)\\

Risk Score of Acquiring COVID-19 & 0.44 (0.72) & 0.22 (0.95) & 0.41 (0.73) & 0.14 (1.09)\\

%Unknown & 0 & 3 &  & \\

Log10 Anti-Nucleocapsid IgG (BAU/ml) &  &  & 0.07 (0.57) & -0.09 (0.52)\\

%Unknown & 5 & 773 & 0 & 3\\

Log10 Anti-Spike IgG (BAU/ml) &  &  & 2.96 (0.41) & 3.04 (0.64)\\

%Unknown & 5 & 769 &  & \\

Log10 Anti-RBD IgG (BAU/ml) &  &  & 3.04 (0.43) & 3.19 (0.65)\\

%Unknown & 5 & 769 &  & \\

Log10 nAb ID50 (IU50/ml) & & & 1.88 (0.59) & 2.08 (0.63)\\
%Unknown & 4 & 769 &  & \\
\hline
\end{tabular}

\caption*{Mean (standard deviation) is reported for quantitative variables and number of participants (percentage) is reported for categorical variables. Age, sex, race and ethnicity, presence of comorbidities, United States Census region, and risk score of acquiring COVID-19 were recorded at randomized placebo-controlled clinical trial enrollment. Testing date and immune marker measures were collected at time of SARS-CoV-2 testing. A Person of Color was defined as a participant who reported their race as American Indian or Alaska Native, Asian, Black or African American, Multiple, Native Hawaiian or Other Pacific Islander, or Other, and/or reported their ethnicity as Hispanic or Latino. A Non-Hispanic/Latino White individual was defined as a participant who reported their race as White and either reported their ethnicity as Not Hispanic or Latino or did not report their ethnicity. 
\\
Abbreviations: Log10 = Logarithm; IgG = Immunoglobulin G; RBD = Receptor-Binding Domain; nAb = Neutralizing Antibody, BAU = Binding Antibody Unit; ml = Milliliter; ID50 = 50\% Inhibitory Dilution; IU50 = 50\% Inhibitory Units
}

\end{table}

The risk score of acquiring COVID-19 was estimated from the placebo arm of COVE \cite{gilbert_immune_2022} and missing for three phase one TND noncases. We converted participants' SARS-CoV-2 testing date into a quantitative calendar date variable using September 1, 2020 as a reference date (for TMLE) or a phase one tertile calendar date variable (for pseudo-likelihood logistic regression and ordinary logistic regression). For the TMLE approach, we standardized risk score of acquiring COVID-19 and calendar date to have mean zero and standard deviation of one.

\begin{table}[!htbp]
\caption{Candidate ensemble super learner algorithms for targeted maximum likelihood estimator to assess COVID-19 immune correlates.}
\label{tab:immsl}
 
    \begin{tabular}{|>{\raggedright\arraybackslash}m{0.5\linewidth}|>{\centering\arraybackslash}p{0.14\linewidth}|>{\centering\arraybackslash}m{0.27\linewidth}|}

    %\toprule
    \hline
          \textbf{Learning Algorithm}&\textbf{sl3 Learner}& \textbf{Algorithm Tuning Parameters and Screens}\\
          \hline
          Sample Mean&Lrnr\_mean& screen = none\\
          Logistic Regression with Main Effects&Lrnr\_glm&  screen = none, lasso\\
          Logistic Regression with Main Effects and Interaction&Lrnr\_glm&  screen = none, lasso\\
          General Additive Models&Lrnr\_gam&  screen = none, lasso\\
           Multivariate Adaptive Regression Splines&Lrnr\_earth& degree = 2; screen = none\\
          Lasso Regression&Lrnr\_glmnet& alpha = 1;  screen = none\\
          Ridge Regression&Lrnr\_glmnet& alpha = 0;  screen = none\\
          %\bottomrule
          \hline
    \end{tabular}
 \caption*{Learners from the sl3 R package \cite{coyle_sl3_2021} were considered for estimating the probability of having a high immune marker level in individuals without SARS-CoV-2 infection, given risk score of acquiring COVID-19, comorbidities, region, and quantitative calendar date and the probability of SARS-CoV-2 infection given risk score of acquiring COVID-19, comorbidities, region, and quantitative calendar date. Risk score of acquiring COVID-19  and calendar date were standardized with mean zero and standard deviation of one. Learners either considered all covariates (no screen) or only covariates with nonzero coefficients based on lasso regression using 10-fold cross-validation (lasso screen). Candidate learners were assessed using a binomial log-likelihood loss and $10$-fold cross validation and combined using non-negative least squares regression. Cross-fitting was used to obtain the targeted maximum likelihood estimator.}
\end{table}
We applied a similar ensemble super learner method as in Section \ref{subsec:appimm}, though we excluded highly adaptive lasso, random forests, and gradient boosting from the library due to the small phase two TND study sample size  (Table \ref{tab:immsl}).
\begin{figure}
    \centering
    \includegraphics[width=1\linewidth]{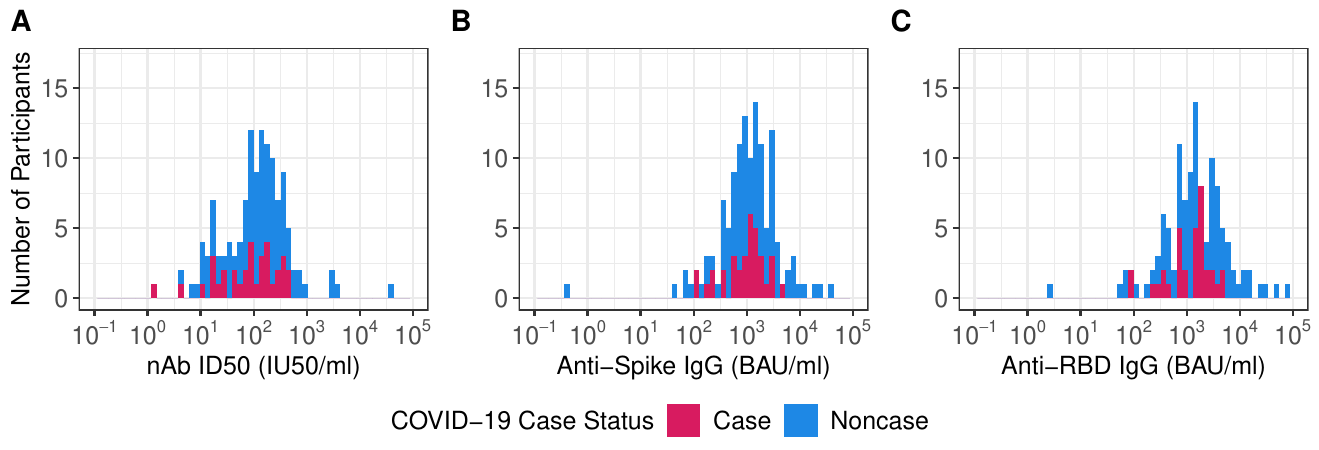}
    \caption{50\% inhibitory dilution neutralizing antibody titer (\textbf{A}), anti-spike binding antibody concentration (\textbf{B}), and  anti-receptor-binding domain binding antibody concentration (\textbf{C}) measured at symptom onset in 127 Moderna COVE phase two test-negative design study participants. Immune marker levels are plotted on the base 10 logarithm scale with plotting labels on the raw scale.
    \\
    Abbreviations: nAb = Neutralizing Antibody; ID50 = 50\% Inhibitory Dilution; IU50 = 50\% Inhibitory Units; IgG = Immunoglobulin G; BAU = Binding Antibody Unit; ml = Milliliter; RBD = Receptor-Binding Domain}
    \label{fig:immcorr}
\end{figure}

\section{Corresponding notation between van der Laan and Gilbert \cite{van_der_laan_semiparametric_2025} and Section \ref{sec:semi}}

    \begin{table}[H]
    \caption{Corresponding notation between the targeted maximum likelihood estimation method described in van der Laan and Gilbert \cite{van_der_laan_semiparametric_2025} and Section \ref{sec:semi} of our manuscript.}
    \label{tab:notation}
   \begin{tabular}{|>{\raggedright\arraybackslash}m{1.9in}|>{\centering\arraybackslash}m{1in}|>{\centering\arraybackslash}m{.7in}|>{\raggedright\arraybackslash}m{2.2in}|}
    \hline
         Variable Description in van der Laan and Gilbert \cite{van_der_laan_semiparametric_2025}& Notation in van der Laan and Gilbert \cite{van_der_laan_semiparametric_2025}&Notation \textcolor{white}{hi} in \textcolor{white}{hi} Section \ref{sec:semi}&Variable Description in Section \ref{sec:semi}\\
         \hline
         Time-Dependent Vaccination Status&  $A(t)$ &$Y, D=1$ &SARS-CoV-2 Infection and Meeting Symptom Definition\\
         [5ex]
         Binary Viral Feature&  $J$ &$A$ &Vaccination Status or Immune Marker Level\\
         [3.5ex]
         Covariates Before Infection&  $W$ &$X$ &Covariates Obtained at TND Enrollment (i.e., Time of Testing)\\
         [5ex]
 Post-Infection Covariates& $W_T$ & -&Not Relevant in TND Context\\
 [2ex]
 Time of First Infection& $T$ & -&Not Relevant in TND Context\\
 [2ex]
         Observed Viral Infection at Monitoring Site&  $R$ &$S$&Observed Individual at TND Testing Site\\
         [3.5ex]
         $J$ Observed in Case-Only Cohort&  $\Delta$ &$\Delta$ &$A$ Observed in TND Study Cohort\\
    \hline
    \end{tabular}
\caption*{van der Laan and Gilbert \cite{van_der_laan_semiparametric_2025} assess relative vaccine efficacy against multiple strains using data from observational case-only studies. Our approach in Section \ref{sec:semi} assesses vaccine effectiveness and immune correlates using data from test-negative design studies.\\
Abbreviation: TND = Test-Negative Design}
\end{table}

\end{document}